\documentclass[letterpaper,twocolumn,10pt,]{article}
\usepackage{usenix-2020-09}

\newcommand{\editf}[1]{{#1}}

\usepackage{tikz}
\usepackage{amsmath,amssymb,amsthm}

\usepackage{filecontents}

\usepackage{enumitem}

\newcommand{\dom}[1]{\mathbb{#1}}
\newcommand{\mech}[1]{\mathcal{#1}}

\def\out[#1]{\dom{Y}_{#1}}
\def\Aldp[#1]{\mathcal{R}_{#1}}

\newtheorem{theorem}{Theorem}[section]

\newtheorem{lemma}[theorem]{Lemma}

\newtheorem{definition}[theorem]{Definition}

\usepackage{pifont}
\newcommand{\yesmark}{\text{\ding{51}}}
\newcommand{\nomark}{\text{\ding{55}}}
\usepackage[linesnumbered,vlined,boxruled,nofillcomment]{algorithm2e}
\DontPrintSemicolon

\usepackage{stmaryrd}
\usepackage{multirow}
\usepackage{balance}

\usepackage[final]{changes}
\makeatletter
\AddToHook{cmd/added/before}{\def\Changes@AuthorColor{blue}}
\AddToHook{cmd/deleted/before}{\def\Changes@AuthorColor{red}}
\AddToHook{cmd/replaced/before}{\def\Changes@AuthorColor{orange}}
\makeatother

\newcommand{\newadded}[1]{{\textcolor{blue}{#1}}}
\renewcommand{\newadded}[1]{{{#1}}}

\usepackage{tcolorbox}
\tcbuselibrary{skins,breakable}

\newtcolorbox{mybox}[2][]{%
  enhanced,
  title        = {#2},
  attach boxed title to top left={xshift=+3mm,yshift*=-3mm},
  colback      = white,
  colframe     = black,
  fonttitle    = \bfseries,
  fontupper    = \small,
  fontlower    = \small,
  colbacktitle = black!3!white,
  coltitle     = black,
  #1
}

\begin{document}

\date{}

\title{\Large \bf Beyond Statistical Estimation:\\ Differentially Private Individual Computation via Shuffling}

\author{
{\rm Shaowei\ Wang$^{1*}$, Changyu\ Dong$^1$, Xiangfu\ Song$^2$, Jin\ Li$^{1,3}$, Zhili\ Zhou$^{1*}$, Di\ Wang$^4$, Han\ Wu$^5$
}\\
\\
$^1$Guangzhou University\ \ \ \ \ \ \ \ $^2$National University of Singapore\\$^3$Guangdong Key Laboratory of Blockchain Security (Guangzhou University)\\
$^4$King Abdullah University of Science and Technology (KAUST)\ \ \ \ \ \ \ \ \ \ $^5$University of Southampton\\
$^*$Corresponding authors: \texttt{wangsw@gzhu.edu.cn}, \texttt{zhou\_zhili@163.com}
} 

\maketitle

\begin{abstract}
In data-driven applications, preserving user privacy while enabling valuable computations remains a critical challenge. Technologies like differential privacy have been pivotal in addressing these concerns. The shuffle model of DP requires no trusted curators and can achieve high utility by leveraging the privacy amplification effect yielded from shuffling. These benefits have led to significant interest in the shuffle model. However, the computation tasks in the shuffle model are limited to statistical estimation, making it inapplicable to real-world scenarios in which each user requires a personalized output. This paper introduces a novel paradigm termed Private Individual Computation (PIC), expanding the shuffle model to support a broader range of permutation-equivariant computations. PIC enables personalized outputs while preserving privacy, and enjoys privacy amplification through shuffling. We propose a concrete protocol that realizes PIC. By using one-time public keys, our protocol enables users to receive their outputs without compromising anonymity, which is essential for privacy amplification. Additionally, we present an optimal randomizer, the Minkowski Response, designed for the PIC model to enhance utility. We formally prove the security and privacy properties of the PIC protocol. Theoretical analysis and empirical evaluations demonstrate PIC's capability in handling non-statistical computation tasks, and the efficacy of PIC and the Minkowski randomizer in achieving superior utility compared to existing solutions.
\end{abstract}

\section{Introduction}\label{sec:intro}
Personal information fuels a wide array of data-driven applications, e.g. statistical analytics, machine learning, recommendation systems, spatial crowdsourcing, e-health, social networks, and smart cities. These applications deliver substantial value but rely on data collected from users, which is a prime target for attacks and carries a high risk of leakage or abuse. Data privacy concerns are escalating, especially after several high-profile data breach incidents. Despite the introduction of stricter privacy laws such as the EU's General Data Protection Regulation and the California Consumer Privacy Act, many users still distrust service providers and are hesitant to consent to the use of their data. To bridge this trust gap and encourage user participation, significant efforts have been made recently to develop privacy-enhancing technologies that enable private data processing.

Two prominent technologies addressing this problem are secure multiparty computation (MPC) \cite{yao1986generate} and differential privacy (DP) \cite{dwork2008differential}. MPC employs interactive cryptographic protocols that enable mutually untrusted parties to jointly compute a function on their private data, ensuring each party receives an output while learning nothing else. Despite its strong privacy guarantees, the cryptographic nature of MPC results in substantial overhead, posing scalability challenges for practical applications. In contrast, DP enhances privacy by adding random noise to the data, allowing computation to be performed on sanitized data without the need for heavy cryptographic machinery. Traditional central DP relies on a trusted curator \cite{dwork2008differential}, but the local model (LDP \cite{kasiviswanathan2011can}) enables users to sanitize their data locally before sharing it. LDP's practicality and minimal trust requirements have led to its adoption by companies such as Apple \cite{tang2017privacy}, Microsoft \cite{ding2017collecting}, and Google \cite{erlingsson2014rappor} in their real-world systems. However, LDP's primary drawback is that each user must independently add sufficient random noise to ensure privacy, which can significantly impact the utility of the data.

Recently, within the realm of DP research, the shuffle model \cite{bittau2017prochlo,erlingsson2019amplification} has emerged. This model introduces a shuffler that randomly permutes messages from users and then sends these anonymized messages to a computing server or analyzer. The server can then compute on these messages to derive the result. Trust in the shuffler is minimized because each user encrypts their locally sanitized data using the server's public key. This way, the shuffler remains oblivious to the messages received from the users. A key advantage of the shuffle model over LDP is utility. It has been proven that anonymizing/shuffling messages amplifies the privacy guarantee provided by the local randomizer used by the users. For instance, shuffled messages from $n$ users each adopting local $\epsilon$-DP actually preserve differential privacy at the level $\epsilon_c = \tilde{O}(\sqrt{e^\epsilon/n})$ \cite{feldman2021hiding,feldman2022stronger}. Consequently, to achieve a predefined global privacy goal, less noise must be added when users sanitize their data locally. This significantly improves the accuracy of the final result. Owing to these utility advantages, extensive studies have been conducted within the shuffle model, \replaced{e.g., \mbox{\cite{balle2020private,ghazi2020private,ghazi2021differentially,girgis2021shuffled}}}{e.g., \mbox{\cite{balle2020private,ghazi2020private,ghazi2021differentially,ghazi2021differentially,girgis2021shuffled}}}.

That said, the shuffle model has a noticeable limitation. The privacy amplification effect relies heavily on anonymizing/shuffling messages, which significantly restricts the types of computation that can be performed. So far, the sole form of computation achievable within the shuffle model is statistical estimation, i.e., the server takes the shuffled messages, aggregates them, and computes a single output from them, e.g., a count, sum, or histogram. However, many real-world applications are non-statistical in nature. When multiple users pool their data together for joint computation, they expect an \emph{individualized} output that may differ for each user. We coined the term ``individual computation'' for such tasks. Examples of individual computation tasks include:
\begin{itemize}[left=0.5em]
\setlength{\labelwidth}{0.1em}
\setlength{\labelsep}{0.3em}
\setlength{\itemsep}{0pt}
\setlength{\parskip}{0pt}
\setlength{\parsep}{0pt}
\item \textbf{Combinatorial optimization:} Spatial crowdsourcing \cite{tong2020spatial}, advertisement allocation \cite{mehta2013online}, and general combinatorial optimization \cite{korte2011combinatorial}, where two or more parties are often matched together based on their private information. Each party should get their own list of ``best matches'' whatever that means.
\item \textbf{Information retrieval:} Mobile search \cite{kamvar2006large}, location-based systems \cite{andres2013geo,cho2011friendship}, where the query results (e.g., nearby restaurants or neighboring users) depend on the private information of the inquirer.
\item \textbf{Incentive mechanisms:} In federated learning \cite{zhan2021survey} or crowd sensing \cite{shah2015double}, incentives play a vital role to encourage well-behaved participation. The amount of rewarding incentive must be computed for individuals based on their contribution (e.g., via Shapley values \cite{roth1988shapley}).
\end{itemize}

At first glance, private individual computation appears unattainable within the shuffle model because the need for personalized output conflicts with the anonymization required for privacy amplification. However, this is not necessarily the case. Our observation is that many individual computing tasks are equivariant to shuffling, meaning that the permutation applied to the inputs does not affect the computation. Therefore, shuffling does not prevent the server from producing personalized answers for each client -- the server does not need to know which answer belongs to whom. Nevertheless, there is a challenge: how to return the output to the correct user without compromising anonymity. One straightforward approach is to have the shuffler maintain a long-term duplex connection channel between each client and the server (e.g., as in an Onion routing network \cite{goldschlag1999onion}). However, this method is costly due to the need to store communication states and may be vulnerable to de-anonymization attacks on anonymous channels \cite{murdoch2005low,overlier2006locating}. Additionally, current shuffler implementations, e.g. the one described in the seminal work \cite{bittau2017prochlo}, do not support such duplex connections. {Moreover, after receiving the computation results, many individual computation tasks require establishing party-to-party communication  (e.g., a user communicates with the matched driver in taxi-hailing services, a user communicates with matched near users in social systems), which can be hard to implement in the duplex shuffle channel.} Addressing this issue is the first technical challenge we need to overcome.

The second challenge we face is designing optimal randomizers for individual computation within the shuffle model. Often, the randomizer in a shuffle model protocol should be tailored to specific tasks. For statistical tasks within the shuffle model, several studies have developed near-optimal randomizers, as demonstrated in histogram estimation \cite{ghazi2021power, feldman2021hiding} and one-dimensional summation estimation \cite{balle2019privacy}. However, the new setting of individual computation is different: the focus is on the accuracy of the output for each user rather than the statistical accuracy of the population. This difference renders existing randomization strategies (e.g., randomizers utilizing dimension sampling, or budget splitting, as reviewed in \cite{xiong2020comprehensive}), as well as prevalent randomizers (such as adding Laplace noise \cite{dwork2008differential}), less effective for the new setting. Therefore, we need to reconsider the fundamental privacy-utility trade-offs and redesign the underlying randomizers to better suit individual computation requirements.


\noindent\textbf{Our contributions} In this work, we introduce a new paradigm extending the shuffle model that allows a wider range of permutation-equivariant tasks to be computed with DP guarantee and can enjoy privacy amplification provided by shuffling. We term the new paradigm \emph{Private Individual Computation} (PIC in short). We define PIC formally as an ideal functionality, which captures its functional and security properties. We also provide a concrete protocol, with formal proof, that can realize the ideal functionality. 

Similar to the shuffle-DP protocols, each user adds noise locally to their data and encapsulates it (and possibly other auxiliary data) into an encrypted message under the public key of the computation server. The shuffler then shuffles the encrypted messages, before sending them to the computation server. The server can decrypt these messages and perform the permutation-equivariant computation. What differs is that each user also includes a one-time public key within the encrypted message. This one-time public key serves two purposes: (1) it allows the server to encrypt the computation result such that it can only be decrypted by the owner of the corresponding private key, and (2) it acts as a pseudonym for the key owner. This approach addresses our first challenge with minimal overhead: the server can publish a list where each entry consists of a public key along with the computation result encrypted under this key. Users can download the entire list and decrypt the entry associated with their own public key, maintaining their anonymity. Additionally, users can establish secure communication channels with other matched parties using their public keys to eventually complete PIC tasks.

Another main contribution is the development of an asymptotically optimal randomizer specifically designed for the PIC model. This randomizer is based on an LDP mechanism we call \emph{Minkowski Response}. The primary goal of this design is to enhance utility, measured by the single-report error, which is the expected squared error between a user's true data value and its sanitized version. To achieve this, a Minkowski distance $r$ is determined based on the privacy budget. The randomizer's output domain is defined as a ball extending the input domain's radius by $r$. The key to achieving high utility lies in the randomizer's output: it selects a value close to the true value (within $r$) with a relatively high probability and a value from elsewhere with a relatively low probability.

We formally prove the security and privacy properties of the PIC protocol. Additionally, we provide a theoretical analysis of the utility bounds achievable by protocols in the PIC model and the Minkowski Response mechanism. Our analysis demonstrates that asymptotically, the error upper bound of the Minkowski Response matches the error lower bound for all possible randomizers in the PIC model, thereby achieving optimality. Alongside theoretical analysis, we conducted extensive experiments using real-world applications and datasets. The evaluation confirms that computations conducted in the PIC model exhibit significantly better utility than those in the LDP model. Furthermore, the performance of the Minkowski randomizer, measured by single-report error and task-specific utility metrics, surpasses that of existing LDP randomizers commonly used in the shuffle model.

\noindent\textbf{Organization.} The remainder of this paper is organized as follows. Section \ref{sec:related} reviews related works. Section \ref{sec:background} provides preliminary knowledge about privacy definitions and security primitives. Section \ref{sec:setting} formalizes the problem setting. Section \ref{sec:protocol} presents the PIC protocol. Section \ref{sec:ldpmechanism} provides optimal randomizers. Section \ref{sec:exp} evaluates the performances of our proposals. Section \ref{sec:dis} discusses more merits and future directions of the PIC. Finally, Section \ref{sec:conclusion} concludes the paper.

\section{Related Works}\label{sec:related}
This section reviews various approaches to private computation, primarily concentrating on non-statistical tasks.

\subsection{Secure Multiparty Computation}
Secure Multiparty Computation (MPC) is a fundamental cryptographic primitive that enables multiple parties to jointly compute a function over their inputs while no party learns anything beyond their own input and the final output of the computation. MPC was first conceptualized by Andrew Yao in the 1980s, and it has been proven that any computable function can be realized by MPC \cite{yao82}. MPC relies on cryptographic protocols that exchange encrypted messages among parties. To allow computing on encrypted data, primitives such as homomorphic encryption \cite{Paillier99, Cheon17}, secret sharing \cite{Bogdanov08,Damgard12}, or garbled circuits \cite{yao1986generate} can be employed. The primary challenge in MPC lies in balancing privacy and efficiency. While MPC offers strong privacy guarantees, it often suffers from significant computational and communication overheads,  which makes scaling to large datasets or numerous parties difficult. Recent research in MPC has been focusing on optimizing protocols and practical implementations \cite{Cramer18, Keller20,Burra21, Dalskov21, RosulekR21, Dalskov22,Smart23, Boyle23}. Despite a significant improvement, MPC still faces efficiency issues that hinder its widespread real-world deployment.

\subsection{Curator and Local DP Methods}
Many works study matching, allocation, or general combinatorial optimization problems within the curator DP model \cite{dwork2008differential} in the presence of a trusted party collecting raw data from clients (e.g., in \cite{mcsherry2009differentially,cormode2012differentially,to2014framework,to2016differentially}). Since the assumption of a trustworthy party is often unrealistic in decentralized settings, many studies adopt the local model of DP (e.g., in \cite{wang2017location,ren2018lopub,to2018privacy,wang2022privacy}), where each client sanitizes data locally and sends the noisy data to the server for executing corresponding matching/allocation algorithms. As each client must injects sufficient noises into data to satisfy local DP, the execution results often maintain low utility.








\subsection{Shuffle Model of DP}
The recently proposed shuffle model \cite{bittau2017prochlo,erlingsson2019amplification} combines the advantages of the curator model (e.g., high utility) and the local model (e.g., minimal trust). Depending on the number of messages each client can send to the intermediate shuffler, the shuffle model can be categorized as single-message \cite{balle2019privacy,erlingsson2019amplification,feldman2022stronger} and multi-message \cite{balle2020private,ghazi2021power}. The single-message shuffle model leverages privacy amplification via shuffling to enhance data utility compared to the local model. A substantial body of work \cite{erlingsson2019amplification,feldman2021hiding,feldman2022stronger} demonstrates that $n$ shuffled messages from clients, each adopting a same $\epsilon$-LDP randomizer, can actually preserve $\tilde{O}(\sqrt{e^\epsilon/n})$-DP. By removing the constraint of sending one message, the multi-message shuffle model can achieve better utility than the single-message model and might be comparable to the curator model (e.g., in \cite{balle2020private,ghazi2021power}). However, each multi-message protocol is tailored to a specific statistical query (e.g., summation), rendering them unsuitable for permutation-equivariant tasks with non-linear computations. {There is a line of works on the shuffle model for private information retrieval (e.g., in \cite{ishai2006cryptography,ishai2024information,gascon2024computationally} with cryptography security and in \cite{toledo2016lower,albab2022batched} with statistical DP), where the query is represented as multiple secret shares before sent to the shuffler, and the server holding the database entries returns linear-transformed entries for each query, using the duplex shuffled communication channel. This kind of duplex-communication shuffle model can be vulnerable to anonymity attacks \cite{murdoch2005low,overlier2006locating}, and is pertained to the linear computation in private information retrieval. It can not be applied to other PIC tasks (such as combinatorial optimization and federated learning with incentive) that involve with non-linear computations, and can not provide secure user-to-user communication needed in tasks like spatial crowdsourcing and social systems.}

Overall, existing works in the shuffle model primarily focus on statistical queries. This work, for the first time, explore the shuffle model for non-statistical applications (i.e., combinatorial optimization, social systems, and incentive mechanisms). 


\subsection{Combining Cryptography and DP}
While cryptographic tools can protect data secrecy during multiparty computation, they do not necessarily preserve output's privacy. DP can be employed to enhance the privacy of the outputting result of secure multiparty computation through decentralized noisy addition \cite{goryczka2015comprehensive}. To account for privacy loss due to intermediate encrypted views in MPC, researchers have proposed the relaxed notion of computational DP \cite{mironov2009computational} against polynomial-time adversaries. Computational DP protocols often inherited the computation/communication complexity of MPC (refer to approaches in Table \ref{tab:mpec}\editf{ of the full version \cite{wang2025PIC}}).

\section{Preliminaries}\label{sec:background}
In this section, we provide a concise introduction to the preliminaries. A list of notations can be found in Table \ref{tab:notations}.
\begin{table}
\caption{List of notations.}
\label{tab:notations}
\centering
\begin{tabular}{|c|l|}
\hline
\bfseries Notation & \bfseries Description\\
\hline
$[i]$ & $\{1,2,...,i\}$ \\
$[i:j]$ & $\{i,i+1,...,j\}$ \\

\hline

$\mech{S}$ & the shuffling procedure\\

$\mech{R}$ & the randomization algorithm\\

$G_i$ & the $i$-th group of users ($i\in [m]$)\\

$n_i$ & the number of users in group $G_i$\\

$u_{i,j}$ & the $j$-th user in group $G_i$ where $j\in [n_i]$\\



\hline
$\epsilon$ & the local privacy budget\\

$\epsilon_c$ & the amplified privacy level\\

\hline

$sk$, $pk$ & the secret key and public key, respectively\\

$\lambda$ & the security parameter of cryptography \\

\hline
\end{tabular}
\end{table}

\subsection{Privacy Definitions}

\begin{definition}[Hockey-stick divergence \cite{sason2016f}]\label{def:hsd}
For two probability distributions $P$ and $Q$, the Hockey-stick divergence between them with parameter $e^\epsilon$ is as follows:
$$D_{\epsilon}(P||Q)=\int_{z\in \mathcal{Z}} \max\{0, P(z)-e^{\epsilon}Q(z)\}\mathrm{d}z.$$
\end{definition}

Differential privacy imposes divergence constraints on output probability distributions with respect to changes in the input. In the curator model of differential privacy, a trusted party collects raw data \(x_i \in \dom{X}\) from all users to form a dataset \(T = \{x_1, \ldots, x_n\}\) and applies a randomization algorithm \(\mech{R}\) to release query results \(\mech{R}(T)\). For two datasets \(T\) and \(T'\) of the same size and differing in only one element, they are referred to as \emph{neighboring datasets}. Differential privacy ensures that the Hockey-stick divergence between \(\mech{R}(T)\) and \(\mech{R}(T')\) is bounded by a sufficiently small value (i.e., \(\delta = O(1/n)\)). 


\begin{definition}[$(\epsilon,\delta)$-DP \cite{dwork2008differential}]\label{def:cdp}
A randomization mechanism $\mech{R}$ satisfies $(\epsilon,\delta)$-differential privacy iff 
$\mech{R}(T)$ and $\mech{R}(T')$ are $(\epsilon,\delta)$-indistinguishable for any neighboring datasets $T,T' \in \dom{X}^n$. That is,  $\max(D_{\epsilon}(\mech{R}(T) \| \mech{R}(T')), D_{\epsilon}(\mech{R}(T') \| \mech{R}(T))) \leq \delta$.
\end{definition}

In the local DP model, each user applies a randomization mechanism $\mech{R}$ to their own data $x_i$, with the objective of ensuring that the Hockey-stick divergence between $\mech{R}(x)$ and $\mech{R}(x')$ is $0$ for any $x,x'\in \dom{X}$ (see Definition \ref{def:ldp}).
\begin{definition}[local $\epsilon$-DP \cite{kasiviswanathan2011can}]\label{def:ldp}
A mechanism $\mech{R}$ satisfies local $\epsilon$-DP iff $D_{\epsilon}(\mech{R}(x)||\mech{R}(x'))=0$ for any $x,x'\in \dom{X}$.
\end{definition}

\subsection{The Classical Shuffle Model}\label{subsec:amplifier}
Following the conventions of the randomize-then-shuffle model \cite{cheu2019distributed,balle2019privacy}, we define a single-message shuffle protocol $\mech{P}$ as a list of algorithms $\mech{P} = (\mech{R}, \mech{A})$, where $\mech{R}: \dom{X} \to \dom{Y}$ is local randomizer on client side, and $\mech{A}: \dom{Y}^n \to \dom{Z}$ is the analyzer on the server side. We refer to $\dom{Y}$ as the protocol's \emph{message space} and $\dom{Z}$ as the \emph{output space}.
The overall protocol implements a mechanism $\mech{P} : \dom{X}^n \to \dom{Z}$ as follows.
User $i$ holds a data record $x_i$ and a local randomizer $\mech{R}$, then computes a message $y_i = \mech{R}(x_i)$.
The messages $y_1,\ldots,y_n$ are shuffled and submitted to the analyzer. We denote the random shuffling step as $\mech{S}(y_1,\ldots,y_n)$, where $\mech{S} : \dom{Y}^n \to \dom{Y}^n$ is a \emph{shuffler} that applies a uniform-random permutation to its inputs.
In summary, the output of $\mech{P}(x_1, \ldots, x_n)$ is represented by $\mech{A} \circ \mech{S} \circ \mech{R}(X) = \mech{A}(\mech{S}(\mech{R}(x_1), \ldots, \mech{R}(x_n)))$.


In particular, when all users adopt an identical $\epsilon$-LDP mechanism $\mech{R}$, recent works \cite{feldman2021hiding,feldman2022stronger} have derived that $n$ shuffled $\epsilon$-LDP messages satisfy $(O((1-e^{-\epsilon})\sqrt{e^{\epsilon} \log(1/\delta)/n}), \delta)$-DP. We denote the amplified privacy level as:
$$\epsilon_c=\texttt{Amplify}(\epsilon,\delta,n),$$ the tight value of which can also be numerically computed \cite{koskela2021tight,wang2023unified}. Our experiments will use numerical bounds in \cite{koskela2021tight}.


\begin{figure*}[h]
\begin{center}
\centerline{\includegraphics[width=150mm]{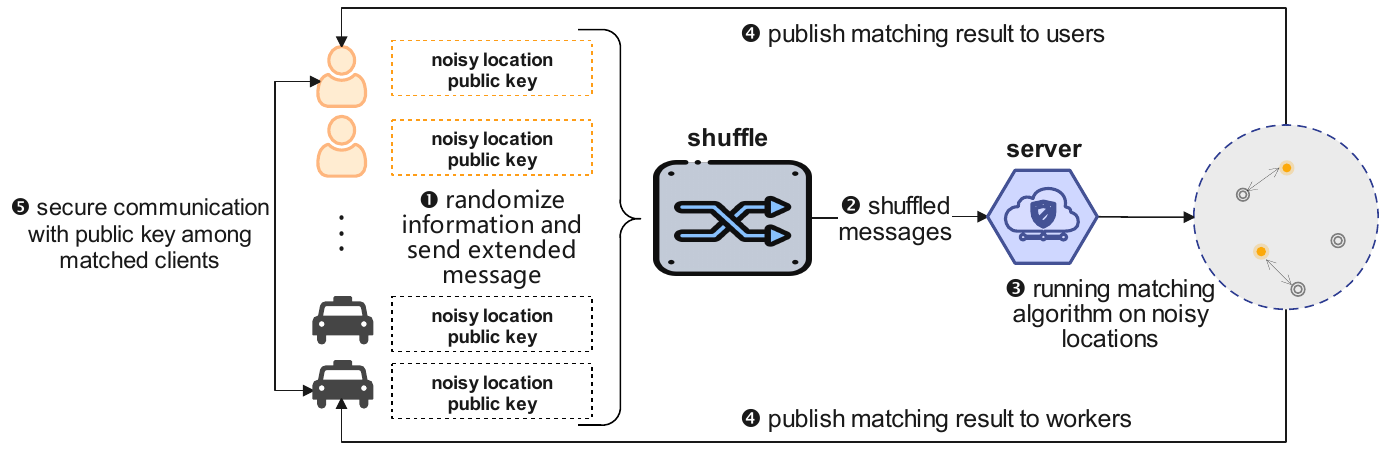}}
\vspace*{-1.0em}
\caption{An illustration of taxi-hailing in the PIC model. Besides (sanitized) location information, each user also encapsulates a one-time random public key into the message to the shuffler.}
\label{fig:overview}
\end{center}
\vspace*{-2.0em}
\end{figure*}

\subsection{Public Key Encryption}
We use an IND-CPA secure public key encryption scheme in our PIC protocol. Formally, a public key encryption scheme $\Pi$ is a tuple of three algorithms $(\textsf{Gen}, \textsf{Enc},\textsf{Dec})$:
\begin{itemize}[left=0em]
\setlength{\labelwidth}{0.1em}
\setlength{\labelsep}{0.2em}
\setlength{\itemsep}{0pt}
\setlength{\parskip}{0pt}
\setlength{\parsep}{0pt}
\item \textsf{Gen}: takes as input the security parameter $\lambda$ and outputs a pair of keys $(pk,sk)$, where $pk$ denotes the public key and $sk$ the private key.
\item \textsf{Enc}: takes as input a public key $pk$ and a message $m$ from the plaintext space. It outputs a  ciphertext $c\rightarrow \textsf{Enc}_{pk}(m)$.
\item \textsf{Dec}: takes as input a private key $sk$ and a ciphertext $c$, and outputs a message $m$ or a special symbol $\bot$ denoting failure.
\end{itemize}

It is required that for every  $(pk,sk)$ and plaintext message $m$, it holds that $\textsf{Dec}_{sk}(\textsf{Enc}_{pk}(m)) = m$. The IND-CPA security ensures the scheme leaks no useful information under a chosen plaintext attack.

\section{Problem Settings}\label{sec:setting}
We now introduce Private Individual Computation, a new paradigm for privacy-preserving computation that offers several benefits: (1) It provides a formal privacy guarantee (in the DP sense) for a wide range of computation tasks. (2) Since the computation is performed on user-sanitized data, it avoids the need for heavy cryptographic protocols, allowing for scalability to handle large user bases. (3) The privacy amplification effect results in significantly better utility compared to local DP. In this section, we will first present a few motivating applications and then formally define Private Individual Computation as an ideal functionality.


\subsection{Motivating Applications}
We describe three prevalent exemplar computation tasks: spatial crowdsourcing, location-based social systems, and federated learning with incentives:

\noindent\textbf{Spatial crowdsourcing.} A spatial crowdsourcing system typically consists of three roles: users, workers, and the orchestrating server. It proceeds through the following major steps:
\begin{itemize}[left=0.5em]
\setlength{\labelwidth}{0.1em}
\setlength{\labelsep}{0.2em}
\setlength{\itemsep}{0pt}
\setlength{\parskip}{0pt}
\setlength{\parsep}{0pt}
\item[I.] {\emph{Task submission:} each user $i\in G_a$ submits a task (e.g., taxi calling requests, sensing requests) $x_i=(i, l_i, v_i)$ containing the location $l_i$ and possibly other information $v_i$;}
\item[II.] {\emph{Worker reporting:} every enrolled worker $j\in G_b$ reports $x_j$ containing the location $l_j$ and other information $v_j$;}
\item[III.] {\emph{Task assignment:} the server receives $\{x_i\}_{i\in G_a}$ and $\{x_j\}_{k\in G_b}$, and outputs a matching $M:G_a\times G_b\mapsto \{0,1\}$ between $G_a$ and $G_b$ based on some criterion (e.g., minimizing total traveling costs, or maximizing matches).}
\item[IV.] {\emph{Task performing:} users and workers retrieve matching results and collaboratively complete the task.} 
\end{itemize}
One example of taxi-haling is illustrated in Figure \ref{fig:overview}.

\noindent\textbf{Location-based social systems.} In a location-based social system, a set of users and a server interact as the following:
\begin{itemize}[left=0.5em]
\setlength{\labelwidth}{0.1em}
\setlength{\labelsep}{0.2em}
\setlength{\itemsep}{0pt}
\setlength{\parskip}{0pt}
\setlength{\parsep}{0pt}
\item[I.] {\emph{Querying:}} each user $i \in G_a$ submits a request $x_i = (i, l_i, v_i)$, which might include location $l_i$ and preferences $v_i$. 
\item[II.] {\emph{Generating recommendation:}} the server receives $\{x_i\}_{i \in G_a}$ and generates a list of recommendations for each user based on specific criteria (e.g., proximity and preferences).
\item[III.] {\emph{Retrieving:}} each user retrieves the recommendations.
\end{itemize}


\noindent\textbf{Federated learning with incentives}.  Federated learning involves a set of users and a server:
\begin{itemize}[left=0.5em]
\setlength{\labelwidth}{0.1em}
\setlength{\labelsep}{0.2em}
\setlength{\itemsep}{0pt}
\setlength{\parskip}{0pt}
\setlength{\parsep}{0pt}
\item[I.] {\emph{Submitting gradient:}} Each user  $i\in G_a$ in each epoch computes an intermediate gradient information $x_i \in [-1,1]^d$ with its local model and data, then submits to the server;
\item[II.] {\emph{Computing incentive:}} The server computes the average gradient $\overline{x} = \frac{1}{n}\sum_{i \in [n]} x_i$. To incentive participation, the server may reward users with monetary tokens (e.g., via cryptocurrency) according to a profit allocation algorithm $V: [-1,1]^{d \times n} \times [-1,1]^d \mapsto [0,1]$.
\item[II.] {\emph{Receiving incentive:}} Each user retrieves the token and claims its monetary incentive.
\end{itemize}


A fundamental distinction between the above applications and those currently studied in the shuffle model is that each participant now expects an output that differs individually, rather than a single collectively aggregated output. Another distinction is that a participant might need to securely communicate with other participants after receiving the individualized computation results, such as the matched user and worker in spatial crowdsourcing will communicate with each other to accomplish the task, and the matched users in location-based social systems would like to securely contact each other afterward. Informally, in such applications, it is necessary to safeguard data privacy so that, aside from the party who generates the data, no one can be certain about that party's data (up to the leakage allowed by differential privacy). To amplify privacy, we also need to maintain anonymity so that for any given message, an adversary (such as the server, an observer, or an unmatched user) should only know that this message comes from a user within a particular group, but nothing more. Following the convention in the shuffle model, the shuffler is trusted to provide anonymity. The shuffler knows the random permutation used in the shuffling process and will not leak it to other parties, although the shuffler is prevented from observing plaintext messages through encryption.

\subsection{The Ideal Functionality}
Following the ideal-real world paradigm, we capture private individual computation formally as an ideal functionality $\mathcal{F}_{\scriptscriptstyle \textsf{PIC}}$, which is shown in Figure \ref{fig:func:k-half-shuffle}.

Essentially, the ideal functionality represents a fully trusted party that interacts with $m$ groups of users ($G_1,\cdots,G_m$) and one server $S$. \added{The adversary controls a collection of corrupted parties $\mathcal{C} \subset G_1 \cup G_2 \cdots G_n \cup \{S\}$. The adversary has full knowledge of the internal states and incoming messages of these corrupted parties.} The ideal functionality sanitizes the input from each user, shuffles the inputs randomly, and then applies a function to compute the output for each user. In the end, each user receives their individual function output, while the server receives a list of the shuffled sanitized user data and a list of the function outputs. It captures the functional requirements of individual computation in the real world: a server performs a computing task using the joint inputs from a set of users (which are sanitized and shuffled), and each user receives an individualized output. It also captures the security requirements: each party receives precisely the specified output, and nothing more. \added{Note that for efficiency, we opt for a centralized version of PIC where a server performs the actual computation over shuffled data. However, a decentralized PIC can also be realized by a group of servers jointly computing the task $f$ through MPC.} 

\noindent{\textbf{Remark 1}} A careful reader may notice that the ideal functionality does not explicitly capture differential privacy. This omission is intentional for the sake of security analysis. In the security analysis, we decouple the privacy requirements into two sets of proofs: the first set demonstrates that our concrete protocol, when executed by real-world parties, realizes the ideal functionality. This means no additional information about parties' inputs is revealed, except for the output given to each party. The second set of proofs establishes that the output given to each party conforms with differential privacy.

\begin{figure}
\centering
    \begin{mybox}{Functionality $\mathcal{F}_{\scriptscriptstyle \textsf{PIC}}$} 
        \textbf{Parameters}: $m \in \mathbb{N}$; $m$ groups of parties $G_1, G_2, \cdots, G_m$, where $n_i = |G_i|$ denotes the number of parties in group $G_i$; the data randomization mechanisms $\mech{R}_i$ for group $G_i$; a server ${S}$.
        
        \textbf{Functionality}: Upon receiving $n = \sum_{i\in [m]} n_i$ inputs $\{x_{i,j}\}_{i \in [m],j \in [|G_i|]}$ from all users, and the description of a function $f$ to be computed over parties' inputs from the server, do the following:    
        \begin{itemize}
            \item Compute $x'_{i,j} \leftarrow R_i(x_{i,j})$ for all $i \in [m]$ and $j \in [|G_i|]$. 
            \item Sample $m$ random permutations $\pi_1, \pi_2, \cdots, \pi_m$, where $\pi_i:[n_i] \rightarrow [n_i]$, perform shuffle over inputs and obtain $L=\left(\{x_{1,\pi_1(j)}'\}_{j\in [n_1]},\cdots ,\{x_{m,\pi_m(j)}'\}_{j\in [n_m]}\right)$, and compute $\left(\{y_{1,\pi_1(j)}\}_{j\in [n_1]},\cdots,\{y_{m,\pi_m(j)}\}_{j\in [n_m]}\right)\leftarrow f(L)$. 
            \item Send $y_{i,\pi_i(j)}$ to party $u_{i,j}$ for all $i \in [m]$ and $j \in n_i$. Additionally, send $(L, f(L))$ to $S$. 
        \end{itemize} 
	\end{mybox}
\vspace*{-1em}
\caption{The functionality $\mathcal{F}_{\scriptscriptstyle \textsf{PIC}}$}\label{fig:func:k-half-shuffle}
\vspace*{-1em}
\end{figure}

\section{A Concrete Protocol}\label{sec:protocol}

\subsection{The Protocol}\label{subsec:protocol}

\replaced{We now present a concrete protocol for PIC in the $\mathcal{F}_{\scriptscriptstyle \textsf{Shuffle}}$-hybrid model. The $\mathcal{F}$-hybrid model in MPC is a conceptual framework that simplifies the design and analysis of secure protocols. In many MPC protocols, parties rely on existing sub-protocols that have already been proven secure for specific tasks. To help protocol designers focus on other aspects of the protocol, these sub-protocols can be abstracted as ideal functionalities, denoted by $\mathcal{F}$, which are assumed to be computed securely by a trusted third party. This results in a hybrid protocol, combining concrete cryptographic operations from the real world with calls to ideal functionalities, thus avoiding the complexity of specifying and analyzing the sub-protocols in detail. In our case, we define the PIC protocol in the $\mathcal{F}_{\scriptscriptstyle \textsf{Shuffle}}$-hybrid model, where $\mathcal{F}_{\scriptscriptstyle \textsf{Shuffle}}$ replaces the secure shuffle sub-protocol.}{We now present a concrete protocol for PIC. Because there are already secure protocols for shuffling, to simplify the description, we present the protocol in the $\mathcal{F}_{\scriptscriptstyle \textsf{Shuffle}}$-hybrid model, in which parties can communicate as usual, and in addition have access to an ideal functionality $\mathcal{F}_{\scriptscriptstyle \textsf{Shuffle}}$ that does the shuffling.} The ideal functionality $\mathcal{F}_{\scriptscriptstyle \textsf{Shuffle}}$ (Figure \ref{fig:shuffle}) is parameterized with a list of parties that are corrupted by the adversary and collude with the server. For those parties, the adversary should know the correspondence between their messages before and after shuffling, hence $\mathcal{F}_{\scriptscriptstyle \textsf{Shuffle}}$ leaks this part of the permutation to the adversary. 

\begin{figure}[htbp]
\centering
    \begin{mybox}{Functionality $\mathcal{F}_{\scriptscriptstyle \textsf{Shuffle}}$} 
        \textbf{Parameters}: $n \in \mathbb{N}$; $n$ parties $P_1, P_2, \cdots, P_n$; a server ${S}$; the corrupted party set $\mathcal{C}$; the leakage $\mathcal{L}(\pi) = \{(i, \pi(i))\}_{i \in [\mathcal{C}]}$ for the permutation $\pi$ being used. 
        
        \textbf{Functionality}: Upon receiving $n$ inputs $\{x_i\}_{i \in [n]}$ from $P_1, P_2, \cdots, P_n$, respectively. 
        \begin{itemize}
            \item Sample a random permutation $\pi \in \mathbf{S}_n$. 
            \item Define $\{y_i\}_{i \in [n]}$ such that $y_i = x_{\pi(i)}$. 
            \item Send $\{y_i\}_{i \in [n]}$ to the server $S$. Additionally if $S\in \mathcal{C}$, send $\mathcal{L}(\pi)$ to \emph{the adversary} $S$.  
        \end{itemize}
	\end{mybox}
\vspace*{-1em}
\caption{The functionality $\mathcal{F}_{\scriptscriptstyle \textsf{Shuffle}}$}\label{fig:shuffle} 
\end{figure} 

The PIC protocol is outlined below:

\begin{enumerate}
\item The server publishes the global parameters, including (1) the specification of a public key encryption scheme $\Pi=(\textsf{Gen}, \textsf{Enc}$, $\textsf{Dec})$, (2) a security parameter $\lambda$, (3) its own public key $pk_c$, generated by invoking $\textsf{Gen}(\lambda)$, (4) for each user groups $G_i$ $(i\in[m])$, a data randomization mechanisms $\mech{R}_i$.

\item We denote the $j$-th user in group $G_i$ as $u_{i,j}$. Each user generates a key pair $(pk_{i,j},sk_{i,j})\leftarrow \textsf{Gen}(\lambda)$. Each user then randomizes their private data $x_{i,j}$ with mechanism $\mech{R}_i$ and obtains $x_{i,j}'\leftarrow \mech{R}_i(x_{i,j})$. Then the sanitized input is concatenated with their own public key, and encrypted with the server's public key $x_{i,j}''\leftarrow\textsf{Enc}_{pk_c}( pk_{i,j}||x'_{i,j})$.

\item Each user in group $G_i$ invokes $\mathcal{F}_{\scriptscriptstyle \textsf{Shuffle}}$ with $x_{i,j}''$, and $\mathcal{F}_{\scriptscriptstyle \textsf{Shuffle}}$ outputs the shuffled messages $\{x_{i,\pi_i(j)}''\}_{j\in [n_i]}\leftarrow\mathcal{S}(\{x_{i,j}''\}_{j\in [n_i]})$ to the server, where $\pi^{-1}_i$ is the (secret) random permutation used during $\mathcal{F}_{\scriptscriptstyle \textsf{Shuffle}}$ for group $G_i$.

\item The server decrypts each set of shuffled messages and obtains a list $L$ for all $m$ groups as: 
\begin{align*}
L=&\left(\{pk_{1,\pi_1(j)}||x_{1,\pi_1(j)}'\}_{j\in [n_1]},\cdots\right.\\
&\ \ \ \left.\{pk_{m,\pi_m(j)}||x_{m,\pi_m(j)}'\}_{j\in [n_m]}\right)
\end{align*}
It then computes the function $f$ over $L$ to produce output for each anonymous user:
$$\left(\{y_{1,\pi_1(j)}\}_{j\in [n_1]},\cdots,\{y_{m,\pi_m(j)}\}_{j\in [n_m]}\right)\leftarrow f(L).$$

\item The server publishes the computation results to a public bulletin board as a list of pairs: $\{(pk_{i,\pi_i(j)},\textsf{Enc}_{pk_{i,\pi_i(j)}}(y_{i,\pi_i(j)})\}_{j\in [n_i]}$, for each group $i\in [m]$.

\item Every user downloads the list, finds in the list the entry with their own public key, and decrypts the payload to get the computation result.
\end{enumerate}

\noindent{\textbf{Remark 2}} In the final step, we adopt the simplest strategy for users to retrieve their results without the server knowing which entry belongs to whom. If bandwidth is a concern, it can be replaced by a more sophisticated (and computationally more expensive) Private Information Retrieval protocol (e.g. \cite{hen23}).

\noindent{\textbf{Remark 3}} To eventually accomplish the spatial crowdsourcing (e.g., taxi-hailing services) or location-based social system tasks, the computation result $y_{i,\pi_i(j)}$ will contain the matched users of the user $\pi_i(j)\in G_i$. That is, $y_{i,\pi_i(j)}$ will encapsulate a list of public keys and noisy location information about the matched users of the user $\pi_i(j)\in G_i$. After receiving the individualized computation results in the final step, each user $\pi_i(j)$ can then securely contact the matched users using the public keys in $y_{i,\pi_i(j)}$ (possibly with the help of a public communication channel, such as a public bulletin). \added{Consider a scenario where user $A$ wants to send a message $p$ (e.g., the precise location for coordinating taxi-hailing services) to a matched user $B$ via a public bulletin board. The process follows this simple protocol: (1) User $A$ posts the tuple $(pk_B, c = \text{Enc}_{pk_B}(p), \text{Sign}_{sk_A}(c))$ on the bulletin board. Here, $\text{Sign}_{sk_A}(\cdot)$ represents a digital signature function using $A$’s private key. (2) User $B$ retrieves the message by using $pk_B$ as an address from the bulletin board, verifies the signature $\text{Sign}_{sk_A}(c)$ using $A$’s public key $pk_A$ to confirm authenticity, and then decrypts $c$ using $B$’s private key to recover the $p$.}

\subsection{Security Analysis}\label{sec:privacy}
This section presents the formal security analysis of the protocol in the previous section. In the following, we will consider a semi-honest adversary who can statically corrupt parties in the protocol. That is, the adversary will faithfully follow the protocol specifications but try to learn more information than allowed through protocol interaction. Also, before running the protocol, the adversary specifies the corrupted parties. The adversary controls the corrupted parties and knows their internal states. We use $\mathcal{C}$ to denote the collection of corrupted parties and $\mathcal{C} \subset G_1 \cup G_2 \cdots G_n \cup \{S\}$. 

We first show that our protocol securely realizes the ideal functionality $\mathcal{F}_{\scriptscriptstyle \textsf{PIC}}$ in the $\mathcal{F}_{\scriptscriptstyle \textsf{Shuffle}}$-hybrid model. This means the protocol leaks no more information than what is allowed by $\mathcal{F}_{\scriptscriptstyle \textsf{Shuffle}}$ and $\mathcal{F}_{\scriptscriptstyle \textsf{PIC}}$. More precisely, each user gets their output from $\mathcal{F}_{\scriptscriptstyle \textsf{PIC}}$, the server gets $(L,f(L))$, and in the case of colluding with some users in $G_i$, the partial permutation $\mathcal{L}(\pi_i)$. Formally we have the following theorem: 

\begin{theorem}[Security]
The Private Individual Computation (PIC) protocol in \S\ref{sec:protocol} securely computes $\mathcal{F}_{\scriptscriptstyle \textsf{PIC}}$ in the $\mathcal{F}_{\scriptscriptstyle \textsf{Shuffle}}$-hybrid model in the presence of any PPT adversary with static corruption.  
\end{theorem}

The proof is simulation-based. It shows that for any PPT adversary $\mathcal{A}$ in the real world, there exists a PPT simulator $\mathcal{S}$ in the ideal world that can generate a simulated view given the corrupted parties' inputs and outputs. Security means that the simulated view is indistinguishable from the view of real-world execution. The detailed proof can be found in Appendix \ref{app:proof}\editf{ of the full version \cite{wang2025PIC}}. 

The above theorem states that the adversary learns strictly no more than the allowed output and leakage by engaging in the protocol execution. Next, we will show how much differential privacy we can get in the presence of such an adversary with such knowledge. We consider an honest user $u_{\scriptscriptstyle i^*,j^*}$, where $i^*\in [m], j^*\in [n_{i^*}]$. At the same time, the set of corrupted users in $ G_{i^*}$ is denoted as $C_{i^*}\subset G_{i^*}$. Differential privacy in our case means that on two neighboring inputs $X=\big(X_1,\cdots,X_{i^*}\cdots,X_m\big)$ and $X'=\big(X_1,\cdots,X'_{i^*}\cdots,X_m\big)$, the output and leakage obtained by the adversary, denoted as $\mathcal{A}(X)$ and $\mathcal{A}(X')$, are close in distribution. Formally, we have the following theorem:

\begin{theorem}[Differential Privacy]\label{the:privacy}
The Private Individual Computation protocol satisfies $(\epsilon_c,\delta)$-DP, i.e.  
\[
\max(D_{\epsilon_c}(\mech{A}(X)||\mech{A}(X')), D_{\epsilon_c}(\mech{A}(X')||\mech{A}(X)))\leq \delta.
\]
In particular, when all users in $G_{i^*}$ use an identical $\epsilon$-LDP mechanism as the data randomizer $\mech{R}_{i^*}$, and for $n'_{i^*}=|G_{i^*} - C_{i^*}|\geq 16 e^\epsilon\log(2/\delta)$ we have:
\begin{equation}\label{eq:cloneamplify}
\epsilon_c=\log\Bigg(1+\frac{e^\epsilon-1}{e^\epsilon+1}\Big(\sqrt{\frac{64 e^{\epsilon} \log 4/\delta}{n'_{i^*}}}+\frac{8 e^\epsilon}{n'_{i^*}}\Big)\Bigg).
\end{equation}
\end{theorem}

The analysis can be divided into two cases: in the first case, $S\in \mathcal{C}$, i.e. S is corrupted. In this case, since the honest user locally randomizes their input, the input enjoys at least $\epsilon$-DP. Then the shuffling will amplify the privacy guarantee. Since the corrupted server receives the partial permutation as the leakage, the amplification depends on the number of uncorrupted users in the same group ($n'_{i^*}$) as the honest user. The amplified $\epsilon_c$ can then be derived following \cite{feldman2021hiding,feldman2022stronger}. In the second case where $S\not\in \mathcal{C}$, the knowledge of the adversary is $y_{i,j}$ for each corrupted user $u_{i,j}\in \mathcal{C}$. The tricky part is that how much information $y_{i,j}$ leaks depends on the function $f$ being evaluated by the server. Hence we consider the worst case where $f$ output $y_{i,j}=L$. Continuing along the same line of thought, we conclude that the level of differential privacy assurance is no less than in the first case. \added{Additionally, since the one-time random secret/public keys are independently and identically distributed, they do not compromise the local privacy guarantees of the sanitized data $x'_{i,j}$ or the privacy amplification effect provided by benign users. As a result, the privacy amplification guarantee in the PIC model can be reduced to that of the classical (single-message) shuffle model. In this context, \cite{feldman2021hiding} provides a similar formula to Equation \ref{eq:cloneamplify} for shuffling $\epsilon$-LDP messages.} The full proof can be found in Appendix \ref{app:privacyproof}\editf{ of the full version \cite{wang2025PIC}}.

\noindent\textbf{Remark 4} 
After decryption, the server obtains $pk_{i,j}||x'_{i,j}$, where $pk_{i,j}$ is a public key not sanitized by the local randomizer. Despite the presence of the public key, $ pk_{i,j}||x'_{i,j}$ and $x'_{i,j}$ are equivalent in terms of privacy amplification. This is because (1) the public key is random and generated independently of $x_{i,j}$, so prepending it to $x'_{i,j}$ does not affect the local DP guarantee—it is the same as $ x'_{i,j} $ itself; and (2) all public keys follow an identical probability distribution across all users in the group $G_i$, ensuring that the privacy amplification effect via shuffling is not degraded\added{ (see formal statement in Lemma \ref{lemma:shuffle1ldp}\editf{ of the full version \cite{wang2025PIC}})}.

\subsection{Discussion on Post-computation Communication}\label{sec::post-computation}
In certain scenarios, such as spatial crowdsourcing and location-based social systems, there may be additional user-to-user communications following the execution of the PIC protocol. For instance, consider a taxi-hailing application where passengers are in a group $G_1$ and taxi drivers in a group $G_2$. After receiving the matching result at the end of the protocol, the passengers must send their locations to the matched drivers, who need to know where to pick them up. On the other hand, the drivers also need to share their identities and locations with the matched passengers. Inevitably, a party has to sacrifice their privacy to the matched parties. 

The post-computation communication may also have privacy implications for other users not in the matched pair. Privacy amplification via shuffling against an adversary relies on the number of parties that remain anonymous. Recall that in Equation \ref{eq:cloneamplify}, the amplified $\epsilon_c$ relies on the number $n'_{i^*}=|G_{i^*} - C_{i^*}|$ of uncorrupted users in a particular group $G_{i^*}$. From the perspective of the adversary, if the post-computation communication compromises the anonymity of an additional set $U$ of the users in $G_{i^*}$, then $n'_{i^*}$ becomes $|G_{i^*} - C_{i^*}- U|$. Accordingly, the privacy amplification effect for users in $G_{i^*}$ that are still anonymous is weakened.\added{ In the specific case of one-to-one matching, we have $|U| \le |C_{i^*}|$ (since some corrupted user may fail to find a match). Therefore, it follows that  $n'_{i^*} \ge |G_{i^*}| - 2|C_{i^*}|$.
}

We emphasize that preventing the loss or weakening of privacy via technical means is not feasible, because the information is necessary for the proper functioning of the application. However, managerial countermeasures, such as ensuring sufficiently large user groups and limiting post-computation exposure according to the need-to-know principle, can mitigate potential privacy risks arising from post-computation communication.

\section{Optimal Randomizers}\label{sec:ldpmechanism}
{
\subsection{Inadequacy of Existing Randomizers}
In PIC, each user first sanitizes their data using an $\epsilon$-LDP randomizer. The design of the randomizer significantly impacts the utility of the tasks. While the PIC model can be seen as an extension of the shuffle model, it has unique characteristics that render the existing LDP randomizers commonly used in the shuffle model inadequate.

The main discrepancy between the shuffle model and the PIC model is that the former emphasizes statistical utility, whereas the latter focuses on the utility of each report. The shuffle model aims to estimate certain statistics from the noisy data collected from users, so it cares about how close the estimation is to the true value of the desired statistic. In the literature, utility is often measured by the expected square error (i.e., the variance) bound of the estimation:
\[
\max_{\scriptscriptstyle T \in \mathbb{X}^n}\mathbb{E}[\|\tilde{f}(T) - f(T)\|_2^2] = \max_{\scriptscriptstyle T \in \mathbb{X}^n}\textsf{Var}[\tilde{f}(T)]
\]
where $f$ is a statistical function, and $\tilde{f}$ is its estimation output by the shuffle protocol. On the other hand, in the PIC model, the tasks are often non-statistical. For example, in location-based matching, the required computation is to take two users' locations and compute the distance between them. Hence, the above utility measure is no longer suitable. It is more natural to measure utility by the single report error:
\[
\max_{x_j \in \dom{X}} \mathbb{E}[\|\mech{R}(x_j) - x_j\|_2^2],
\]
where $\mech{R}$ is the LDP-randomizer employed by the users.\added{ Specifically, the error of addition/minus among $k$ reports (e.g., distance between two reports) can also be upper bounded by $k$ times of the single report error.} Additionally, in the PIC model, user data is typically multi-dimensional (e.g., location, gradient), making sanitization significantly more difficult compared to scalar data. Another notable characteristic of the PIC model is that the local differential privacy budget $\epsilon$ is relatively large, often scaling linearly with $\log(n'_{i^*})$, as implied by Theorem \ref{the:privacy}. These factors together create issues when existing LDP randomizers are applied directly in the PIC model.

To understand the problem, we first examine a class of LDP randomizers \cite{nguyen2016collecting, wang2019collecting, liu2020fedsel, jiang2022signds} that operate by sampling a few dimensions from $[d]$. Each user submits an incomplete report that contains only the sampled dimensions (with added noise) from their local data. On the positive side, this strategy reduces the amount of noise added to the sampled dimensions. On the negative side, the unsampled dimensions are missing. In statistical estimation tasks, the estimation is made using all reports, each covering some dimensions. Therefore, the incompleteness of a single report is less important, and better utility can be achieved. However, in the PIC model, where the focus shifts to the error of individual reports, this strategy may lead to worse results.

There have been LDP randomizers that submit complete reports. One obvious strategy is to explicitly split the local budget into $d$ parts and then apply a one-dimensional LDP mechanism independently to each dimension, or implicitly distribute the budget among dimensions, as seen in the Laplace \cite{dwork2006differential}, PlanarLaplace \cite{andres2013geo}, PrivUnit \cite{bhowmick2018protection}, and PrivUnitG mechanisms \cite{asi2022optimal}. However, these approaches are sub-optimal in the high budget regime. Specifically, even if we use the optimal one-dimension randomizer \cite{balle2019privacy}, splitting the budget across each dimension and then applying any randomizer for each dimension will result in a mean squared error (MSE) of at least $\frac{d}{(e^{\epsilon/d}-1)^{2/3}}$. The Laplace/PlanarLaplace mechanisms introduce an MSE rate of $\frac{d}{\epsilon^2}$, while the PrivUnit/PrivUnitG mechanisms incur an MSE rate of $\frac{d}{\min\{\epsilon,\epsilon^2\}}$. In contrast, later we will show that the MSE rate can be improved to $(e^{\epsilon}-1)^{-2/(d+2)}$ (see Section \ref{subsec:bounds} and Appendix \ref{app:errorupper}).

Another strategy for submitting complete reports involves using additional randomization techniques such as random projection, data sketches, public randomness, or quantization. Randomizers employing this strategy \cite{smith2017interaction,feldman2021lossless,shah2022optimal,chen2022breaking,isik2024exact,asi2024fast} avoid the issue of incomplete reports but introduce additional noise because of the extra randomization. While the additional noise is not significant in the low budget regime (i.e., $\epsilon = O(1)$) that is typical in the LDP model, it becomes dominant in the PIC model where the local budget can be as large as $\tilde{O}(\log n'_{i^*})$. The resulting additional error will never diminish even when $\epsilon \to +\infty$.

\subsection{Randomizer Design}\label{subsec:mechanism}

We now introduce an asymptotically optimal randomizer, tailored for the PIC model. The randomizer uses an LDP mechanism, which we termed as \emph{Minkowski Response}. For ease of presentation, here we will focus primarily on the $\ell_2$ case (and the $\ell_{+\infty}$ case in Appendix \ref{app:errorupper}), but the mechanism can be generalized to other Minkowski distances.

Without loss of generality, we assume the user data domain to be an $\ell_{2}$-bounded hyperball  $\dom{X} = \{x\ |\ x\in \mathbb{R}^d\ and\ \|x\|_2\leq 1\}$. Most, if not all, real-world data domains can be normalized to $\dom{X}$ (e.g. gradient vector, set-valued data, and location data). We also denote an $\ell_2$-bounded hyperball with radius $r$ centered at any $x\in \mathbb{R}^d$ as follows:
\[\mathbb{B}_{r}(x)=\{x'\ |\ x'\in \mathbb{R}^d\ and\ \|x'-x\|_2\leq r\},\]
and it is shorted as $\mathbb{B}_{r}$ when $x=\vec{0}$.

Minkowski Response works by first defining a distance $r$ based on the local privacy budget as the following:
$$r=\big((e^{\epsilon}-1)^{1/(d+2)}-1\big)^{-1}.$$
Then given the input domain $\dom{X}$, the output domain $\dom{Y}_{r}$ is defined by expanding $\dom{X}$ by $r$:
\[\dom{Y}_{r}=\{y\ |\  y\in \mathbb{R}^d\ and\ \exists x\in \dom{X}\ that\ y\in \mathbb{B}_{r}(x)\}.\] 

\begin{figure}[tb]
\vspace*{-2em}
\begin{center}
\centerline{\includegraphics[width=85mm]{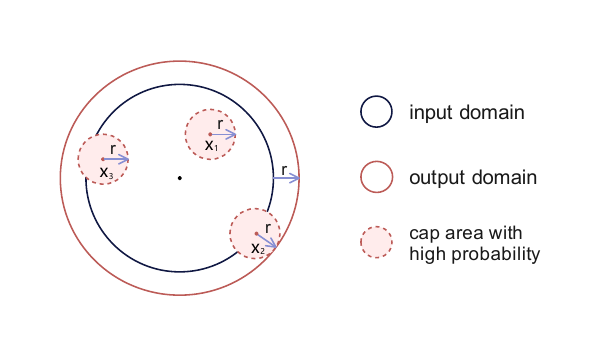}}
\vspace*{-3.0em}
\caption{The probability design of Minkowski response mechanism with a radius $r$.  Illustrated are three inputs $x_1, x_2$ and $x_3$, along with their respective cap areas.}
\label{fig:minkowski}
\end{center}
\vspace*{-2.5em}
\end{figure}

For any input $x\in \dom{X}$, Minkowski Response outputs an output $y\in \dom{Y}_{r}$ with relatively high probability in the cap area $\mathbb{B}_{r}(x)$ and relatively low probability in remaining domain $\dom{Y}_{r}\backslash \mathbb{B}_{r}(x)$ (\added{insipired by extremal probability design \cite{kairouz2016extremal} and geometric-based noise schemes \cite{fitzsimons2024private,lebeda2024better}, }see Figure \ref{fig:minkowski}). Formally:
\begin{equation}\label{eq:mr}
    y= \left\{
    \begin{array}{@{}ll@{}} 
    \text{uniform}(\mathbb{B}_{r}(x)), &\text{with prob. $\frac{\mathrm{V}(\mathbb{B}_{r})\cdot (e^\epsilon-1)}{\mathrm{V}(\dom{Y}_{r})+\mathrm{V}(\mathbb{B}_{r})\cdot (e^\epsilon-1)}$;}\\
    \text{uniform}(\dom{Y}_{r}), &\text{with prob. $\frac{\mathrm{V}(\dom{Y}_{r})}{\mathrm{V}(\dom{Y}_{r})+\mathrm{V}(\mathbb{B}_{r})\cdot (e^\epsilon-1)}$,}\\
    \end{array} \right. 
\end{equation}
where $V(*)$ denote the volume of the corresponding domain.

Lastly, $y$ is debiased to $\tilde{x}$ so that $\mathbb{E}[\tilde{x}]={x}$ as follows:
\begin{equation}\label{eq:debias}
\tilde{x}=y\cdot\frac{\mathrm{V}(\dom{Y}_{r})+\mathrm{V}(\mathbb{B}_{r})\cdot (e^\epsilon-1)}{V(\mathbb{B}_{r})\cdot (e^\epsilon-1)}.
\end{equation}

 It is obvious that the output of Minkowski Response is informative in every dimension. Therefore it avoids problems brought up by incomplete reports. Also, intuitively it has a better utility because the mechanism is more likely to output a value near the true value (i.e., in $\mathbb{B}_{r}(x)$), than from other parts of the output domain (i.e. $\dom{Y}_{r}\backslash \mathbb{B}_{r}(x)$). When the budget $\epsilon$ gets large, the error rate of Minkowski Response decays faster than in previous LDP mechanisms. More specifically, the decay rate of Minkowski Response is $(e^{\epsilon}-1)^{-2/(d+2)}$ (Equation \ref{eq:ldpmse} in Appendix \ref{app:errorupper}), while that of the previous mechanisms is $d/(e^{\epsilon/d}-1)^{2/3}$ or $d/\epsilon^2$. Therefore, the utility advantage of Minkowski Response becomes more significant when $\epsilon$ gets larger. When $n'_{i^*}\to +\infty$ (and thus $\epsilon\to +\infty$), $r$ becomes 0, and the error goes to zero (i.e. no additional error).

\subsection{Analysis of Minkowski Response}\label{subsec:bounds}
The local privacy guarantee of the randomizer is presented in Theorem \ref{the:ldp}.


\begin{theorem}[Local Privacy Guarantee]\label{the:ldp}
Given input domain $\dom{X}=\mathbb{B}_{1}$, the Minkowski response mechanism defined in Equation \ref{eq:mr} satisfies $\epsilon$-LDP.    
\end{theorem}
\begin{proof}
It is observed that the output probability distribution in Equation \ref{eq:mr} is valid for any input $x\in \dom{X}$, the probability density in the cap area $\dom{B}_{r}(x)$ is $\frac{e^\epsilon}{\mathrm{V}(\dom{Y}_{r})+\mathrm{V}(\mathbb{B}_{r})\cdot (e^\epsilon-1)}$, and the density in the non-cap area $\dom{Y}_r\backslash\dom{B}_{r}(x)$ is $\frac{1}{\mathrm{V}(\dom{Y}_{r})+\mathrm{V}(\mathbb{B}_{r})\cdot (e^\epsilon-1)}$. Therefore, for any $x,x'\in \dom{X}$, we have $\frac{\mathbb{P}[\mech{R}(x)=y]}{\mathbb{P}[\mech{R}(x')=y]}\leq e^\epsilon$ for all possible $y\in \dom{Y}_r$, establishing the local $\epsilon$-DP guarantee of the Minkowski response mechanism $\mech{R}$.     
\end{proof}

Next, we analyze the utility of the Minkowski Response in the PIC model. As previously mentioned, in the PIC model, the single report error is a more appropriate measure of utility compared to the statistical errors used in the conventional shuffle model. In Theorem \ref{the:errorlower}, we examine the single report error in the PIC model and establish its lower bound. \added{Essentially, shuffling and its privacy amplification effects complicate error lower bounding in PIC model given privacy constraints, when compared to local DP settings (e.g., in \cite{bhowmick2018protection,duchi2018minimax,duchi2019lower}). Fortunately, the global differential privacy restricts the probability ratio of any event observed in the shuffled messages when given two neighboring input datasets \cite{cheu2019distributed}. Therefore, inspired by the tight lower bounding procedure for one-dimensional data in \cite{balle2019privacy}, we firstly discrete the mutli-dimensional input/output domain and decompose the single report error into two parts: the one due to the probability of reporting other values than a certain input (i.e., one minus the true positive rate), and the other due to probabilities of reporting a certain value when given other input values (i.e., the false positive rate). We then establish a connection between the global DP parameters, and errors due to \emph{one minus the true positive rate} \& \emph{the false positive rate}. Finally, we show that at least one of these two errors must be as large as $1/n^{\frac{2}{d+2}}$.} The detailed proof is provided in Appendix \ref{app:errorlower}.

\begin{theorem}[Error Lower Bounds]\label{the:errorlower}
Given \added{fixed} \replaced{$d\in \mathbb{N}^+$}{$d\in \mathbb{N}$}, $\epsilon_c>0$, $\delta\in (0,0.5]$, $\dom{X}=\mathbb{B}_{1}(\{0\}^d)$, then for any randomizer $\mech{R}:\dom{X}\mapsto \mathbb{R}^{d'}$ such that $\mech{S}\circ\mech{R}(X)$ and $\mech{S}\circ\mech{R}(X')$ are $(\epsilon_c,\delta)$-indistinguishable for all possible neighboring datasets $X,X'\in \dom{X}^n$, and for any estimator $f:\mathbb{R}^{d'}\mapsto \mathbb{R}^{d}$, we have: $$\max_{x\in \dom{X}}\mathbb{E}\big[\|f\circ\mech{R}(x)-x\|_2^2\big] \geq \tilde{\Omega}\big({1}/{n^{\frac{2}{d+2}}}\big).$$
\end{theorem}

The above theorem suggests that in the PIC model\added{ with fixed privacy requirements of $(\epsilon,\delta)$ independent of group size $n$}, for any randomizer, the single report error is at least $\tilde{\Omega}\big({1}/{n^{\frac{2}{d+2}}}\big)$. Clearly, a randomizer offers better utility if its error is closer to this bound. For the randomizer using Minkowski Response, as presented in section \ref{subsec:mechanism}, we can demonstrate that its single report error has an upper bound. This is summarized in Theorem \ref{the:errorupper}, with the proof provided in Appendix \ref{app:errorupper}.

\begin{theorem}[Error Upper Bounds]\label{the:errorupper}
For any $d\in \mathbb{N}$, $\epsilon_c>0$, $\delta\in (0,0.5]$, $\dom{X}=\mathbb{B}_{1}(\{0\}^d)$, if $\epsilon_c\leq O(1)$ and $n> \max\{16\log(1/\delta),\frac{2^{d+7}\log(1/\delta)}{(e^{\epsilon_c}-1)^2}\}$, then there exist a randomizer $\mech{R}:\dom{X}\mapsto \mathbb{R}^d$ such that $\mech{S}\circ\mech{R}(X)$ and $\mech{S}\circ\mech{R}(X')$ are $(\epsilon_c,\delta)$-indistinguishable for all possible neighboring datasets $X,X'\in \dom{X}^n$, and:
\[\max_{x\in \dom{X}} \mathbb{E}\big[\|\mech{R}(x)-x\|_2^2\big] \leq O\Bigg(\Big(\frac{\log(1/\delta)}{n \epsilon_c^2 }\Big)^\frac{2}{d+2}\Bigg).\]
\end{theorem}


We observe that in most applications, $(\epsilon_c, \delta)$ are given as fixed system parameters. If we consider $(\epsilon_c, \delta)$ as constants, then the error upper bound of the Minkowski Response randomizer in Theorem \ref{the:errorupper} is $\tilde{O}\left(\frac{1}{n^{2/(d+2)}}\right)$, which matches the error lower bound of the PIC model in Theorem \ref{the:errorlower}. This implies that the utility of the Minkowski Response is asymptotically optimal. In contrast, using existing LDP randomizers in the PIC model results in a larger single report error of $\tilde{O}\left(\frac{d}{n^{2/(3d)}}\right)$ or $\tilde{O}\left(\frac{d}{\log^2 n}\right)$ (note that $d > 1$ and $n$ is often not small). Although asymptotic notations describe behavior as $n \rightarrow \infty$, in practice, the utility advantage of the Minkowski Response becomes noticeable without $n$ being very large: in our experiments, the Minkowski Response randomizer outperforms existing LDP randomizers in the PIC model once $n$ reaches the order of $10^2$. If we compare Minkowski Response in the PIC model to using LDP directly (without shuffling), the utility advantage is even greater: any randomizers in the LDP model must endure a single report error of $\Omega\left(\frac{d}{\epsilon_c^2}\right)$ when $\epsilon_c \leq O(1)$ \cite{bhowmick2018protection,duchi2019lower}.

\section{Experimental Evaluation}\label{sec:exp}
{
We evaluate the efficacy of our PIC protocol and Minkowski randomizer. We compare the utility of our proposal against state-of-the-art works in the context of three representative individual computation tasks: spatial crowdsourcing, location-based social systems, and federated learning with incentives.

\subsection{Spatial Crowdsourcing}\label{sec:exp:crowdsouring}
\noindent\textbf{Datasets} We use two real-world datasets: GMission dataset \cite{chen2014gmission} for scientific simulation, and EverySender dataset \cite{tong2016online} for campus-based micro-task completion. Details about the two datasets are summarized in Table \ref{tab:scdatasets}, including the number of users/workers, location domain range, and serving radius of workers about these two datasets (more information can be found in Appendix \ref{app:expSC}\editf{ of the full version \cite{wang2025PIC}}). We normalize location data to the domain $[-1,1]\times[-1,1]$ and scale the serving radius to $1.0\cdot\frac{1-(-1)}{5-0}=0.4$ correspondingly.

\begin{table}[h]
\vspace*{-1.0em}
\caption{Summary statistics of spatial crowdsourcing datasets.}
\setlength{\tabcolsep}{0.3em}
\renewcommand{\arraystretch}{1.1}
\label{tab:scdatasets}
\small
\centering
\begin{tabular}{c|c|c|c|c}
\hline
\textbf{Dataset} & \textbf{users} & \textbf{workers} & \textbf{location domain} & \textbf{serving radius} \\ \hline
GMission         & 713                    & 532                      & $[0,5.0]\times[0,5.0]$   & $1.0$                 \\ \hline
EverySender      & 4036                   & 817                      & $[0,5.0]\times[0,5.0]$   & $1.0$                 \\ \hline
\end{tabular}
\end{table}

\noindent\textbf{LDP randomizers} We use the Minkowski response for location randomization. As a comparison, in the local model of DP (e.g., in \cite{wang2017location,to2018privacy,wang2022privacy}), we compare with existing mechanisms including Laplace  \cite{dwork2008differential}, Staircase \cite{geng2015staircase}, PlanarLaplace with geo-indistinguishability \cite{andres2013geo}, SquareWave \cite{li2020estimating}, PrivUnit \cite{bhowmick2018protection} and its Gaussian variant PrivUnitG in \cite{asi2022optimal}\added{ (see Appendix \ref{app:mechdetail} for implementation details)}. 

\noindent\textbf{Server-side algorithms} Two commonly used server-side algorithms are evaluated as concrete tasks: minimum weighted full matching \cite{jonker1988shortest} and maximum matching \cite{hopcroft1973n}. The two algorithms have different optimization objectives, thus later we will show the results for task-specific utility for each of them in addition to single report errors. The minimum weighted full matching aims to minimize the overall traveling costs between users and workers. The maximum matching aims to maximize the number of successfully matched user/worker pairs, where users outside the workers' serving radius ($0.4$) are deemed unreachable.

\noindent\textbf{Experimental results} We first present the single report errors, quantified by the expected $\ell_2$ distance between the reported location and the true location. In Fig. \ref{fig:SCdistance} (a) and (b), we compare the single report errors within the LDP model. The Minkowski randomizer's error is on par with other state-of-the-art mechanisms when $\epsilon \leq 1.0$ and significantly lower when $\epsilon \geq 2.0$. Fig. \ref{fig:SCdistance} (c) and (d) illustrate the single report errors in the PIC model. For both datasets, the Minkowski randomizer performs better, with its utility advantage increasing as the global privacy budget $\epsilon_c$ grows. Additionally, it is evident that the error is generally higher for all randomizers in the LDP model compared to the PIC model. This discrepancy is due to the lack of privacy amplification in the LDP model, highlighting the benefits of employing the PIC model.

\begin{figure}[tb]
\begin{center}
\centerline{\includegraphics[width=83mm]{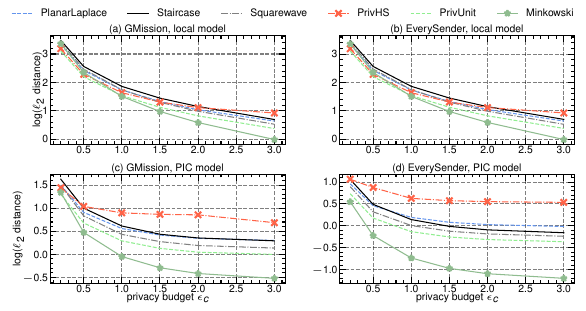}}
\vspace*{-1.5em}
\caption{Expected $\ell_2$ distances of reported locations to true locations on GMission and EverySender dataset.}
\label{fig:SCdistance}
\end{center}
\vspace*{-3.0em}
\end{figure}

Next, we compare the utility when using the minimum weight matching algorithm, as shown in Fig. \ref{fig:SCcost}. Task-specific utility is evaluated by the total travel costs, which is the sum of the actual Euclidean distances between all matched users/workers:
\[
\sum\nolimits_{(i,j)\in G_a\times G_b} \llbracket M(i,j)>0 \rrbracket \cdot \|l_i - l_j\|_2,
\]
where $\llbracket * \rrbracket$ denotes the Iverson bracket. It is observed that although some randomizers, like PrivUnit, exhibit good single-report errors, their task-specific utility is not as favorable. Conversely, the Minkowski randomizer shows consistent performance, outperforming the others in this comparison.

\begin{figure}[!htb]
\begin{center}
\centerline{\includegraphics[width=83mm]{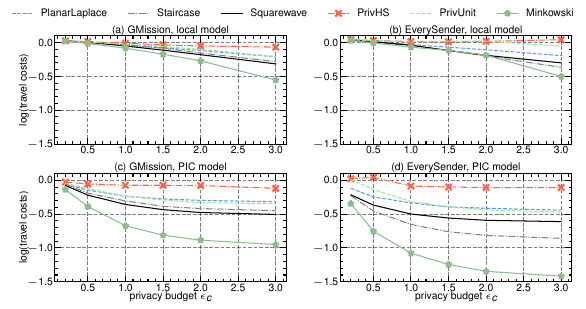}}
\vspace*{-1.5em}
\caption{Travel costs of minimum weighted matching on GMission and EverySender dataset.}
\label{fig:SCcost}
\end{center}
\vspace*{-2.0em}
\end{figure}

Finally, we present the utility comparison using the maximum matching algorithm, as shown in Fig. \ref{fig:SCsuccess}. In this case, utility is assessed by the successful matching ratio:
\[
\frac{\sum_{(i,j)\in G_a \times G_b} \llbracket M(i,j)>0 \rrbracket \cdot \llbracket \|l_i - l_j\|_2 \leq \tau \rrbracket}{\min\{|G_a|, |G_b|\}}.
\]
For both datasets, the matching ratio over clear data is 100\%. The figure demonstrates that the PIC model enhances the matching ratio for all randomizers due to privacy amplification effects. The Minkowski randomizer in the PIC model significantly outperforms the others and the matching ratio approaches an acceptable level for practical use with a reasonable degree of privacy protection.

\begin{figure}
\begin{center}
\centerline{\includegraphics[width=83mm]{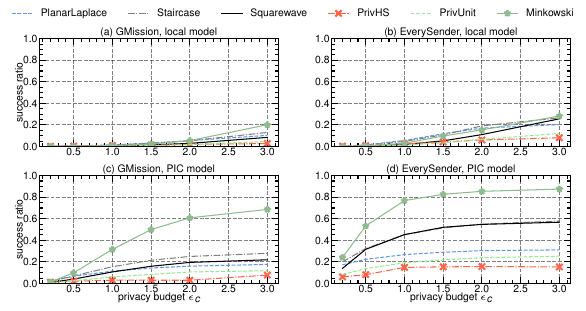}}
\vspace*{-1.5em}
\caption{Success ratios of maximum matching in spatial crowdsourcing on GMission and EverySender dataset.}
\label{fig:SCsuccess}
\end{center}
\vspace*{-2.0em}
\end{figure}

\subsection{Location-based Social Systems}
\begin{table}[ht]
\vspace*{-1.5em}
\caption{Details of location social network datasets.}
\setlength{\tabcolsep}{0.05em}
\renewcommand{\arraystretch}{1.1}
\label{tab:lbsdatasets}
\small
\begin{tabular}{c|c|c}
\hline
\textbf{Dataset} & \textbf{check-ins} & \textbf{location domain}\\ \hline
Gowalla{ \scriptsize(San Francisco)}         & 138368                                  & $[37.54, 37.79]$$\times$$[-122.51, -122.38]$                \\ \hline
Foursquare{ \scriptsize(New York)}                 & 227428                      & $[40.55, 40.99]$$\times$$[-74.27, -73.68]$                \\ \hline
\end{tabular}
\end{table}
\noindent\textbf{Datasets} We use two real-world datasets: Gowalla dataset \cite{cho2011friendship} and Foursquare dataset \cite{yang2014modeling}.  Gowalla and Foursquare are location-based social network websites where users share their locations by checking-in. Details about the two datasets are summarized in Table \ref{tab:lbsdatasets}. As before, the location data is normalized to  $[-1,1]\times[-1,1]$.

\noindent\textbf{LDP randomizers} The LDP randomizers used are the same as those in Section \ref{sec:exp:crowdsouring}.

\noindent\textbf{Server-side algorithm} The server performs the radius-based nearest neighbor (NN) search for the users, which is a common task in location-based social networks \cite{cho2011friendship}. We set the search radius $\tau$ to 0.2, so that each user has several hundreds or thousands of neighbors, varying due to check-in densities. Note that in this application scenario, the returned neighbors may be deanonymized in the post-computation phase. Hence, the actual privacy guarantee in the PIC model depends on the number of users who remain anonymous (see discussion in Section \ref{sec::post-computation}). Considering this, and each user normally has about $n\cdot\frac{\pi\cdot\tau^2}{2^2}\approx n\cdot 3.2\%$ neighbors, we use privacy amplification population $\lfloor n\cdot 90\%\rfloor$ instead of the group size $n$ when calculating the local budget given the global $\epsilon_c$.

\begin{figure}[!htb]
\begin{center}
\centerline{\includegraphics[width=83mm]{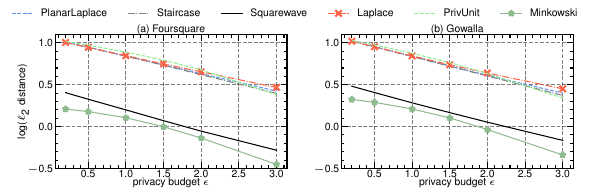}}
\vspace*{-1.5em}
\caption{Expected $\ell_2$ distances of reported locations in location-based social systems with the local model of DP.}
\label{fig:LBSdistancelocal}
\end{center}
\vspace*{-2.0em}
\end{figure}

\begin{figure}[!htb]
\begin{center}
\centerline{\includegraphics[width=83mm]{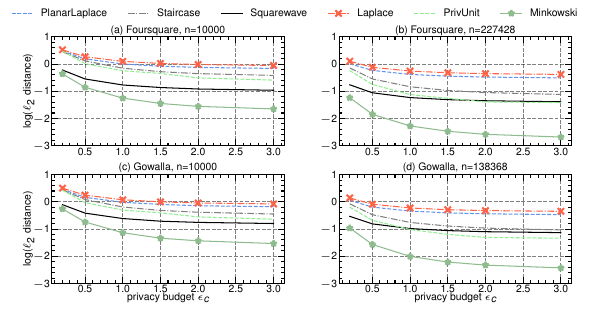}}
\vspace*{-1.5em}
\caption{Expected $\ell_2$ distances of reported locations in location-based social systems with PIC model.}
\label{fig:LBSdistanceshuffle}
\end{center}
\vspace*{-2.0em}
\end{figure}

\noindent\textbf{Experimental results} We first present the single report errors in the LDP model (Fig. \ref{fig:LBSdistancelocal}) and the PIC model (Fig. \ref{fig:LBSdistanceshuffle}). In our experiments, each check-in report is treated as if it were submitted by a separate user, with all users in the same group. We tested two scenarios: one with a random subset of 10,000 check-ins and the other using all check-ins. On both datasets, PIC demonstrates better utility than LDP, and the Minkowski randomizer consistently performs the best across all settings. Additionally, it is evident that the number of anonymous users is a crucial parameter, significantly impacting utility.

\begin{figure}[!htb]
\begin{center}
\centerline{\includegraphics[width=83mm]{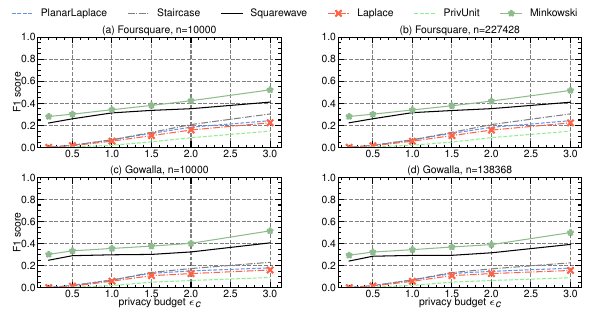}}
\vspace*{-1.5em}
\caption{F1 scores of nearest neighbor queries (LDP model).}
\label{fig:LBSf1local}
\end{center}
\vspace*{-2.0em}
\end{figure}

\begin{figure}[!htb]
\begin{center}
\centerline{\includegraphics[width=83mm]{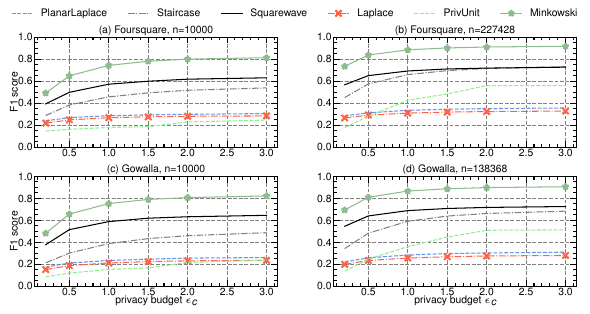}}
\vspace*{-1.5em}
\caption{F1 scores of nearest neighbor queries (PIC model).}
\label{fig:LBSf1shuffle}
\end{center}
\vspace*{-2.0em}
\end{figure}

The task-specific utility metric we use for nearest neighbor queries is the $F_1$ score. Let $N_i$ denote the set of true neighbors of user $i$ within an $\ell_2$-distance $\tau$, and let $\widehat{N}_i$ denote the retrieved neighbor set computed using the sanitized reports. The precision of nearest neighbor queries is defined as 
$
{\sum_{i\in [n]} |N_i \cap \widehat{N}_i|}/{\sum_{i\in [n]} |\widehat{N}_i|}
$,
and the recall is defined as 
$
{\sum_{i\in [n]} |N_i \cap \widehat{N}_i|}/{\sum_{i\in [n]} |N_i|}
$.
The $F_1$ score is as:
\[
F_1 \text{ score} = 2 \cdot \frac{\text{precision} \cdot \text{recall}}{\text{precision} + \text{recall}}.
\]

In the LDP model (Fig. \ref{fig:LBSf1local}), the $F_1$ score is low, even with a large $\epsilon$. While the group size affects the number of neighbors each user has, it does not impact the $F_1$ score. Conversely, in the PIC model (Fig. \ref{fig:LBSf1shuffle}), the $F_1$ score is significantly higher. With a large group, the Minkowski randomizer can achieve an $F_1$ score of 0.9 with a stringent global budget $\epsilon_c=1$ and nearly 1 if the budget is loosened to $\epsilon_c=3$.

\subsection{Federated Learning with Incentives}
In this application, the users collaboratively train a model by Federated learning, and the server decides each user's incentive by how much their local gradient contributes to the global model. To deliver monetary incentive rewards, we assume the use of an untraceable cryptocurrency (e.g., ZCash \cite{hopwood2016zcash}). Each user can  generate an additional public/privacy key pair according to the specification of the cryptocurrency, derive a wallet address from the public key, and append the wallet address to the end of their report. Based on the Shapley values, the server determines the monetary incentives and distributes them to the users using appended wallet addresses.

\noindent\textbf{Datasets} We utilize the MNIST handwritten digit dataset to train a simple neural network using federated learning. The MNIST dataset comprises 60,000 images, with 50,000 designated as training samples. Each user receives one training sample and trains a neural network model, as detailed in Table \ref{tab:cnn}, which contains $d = 4292$ trainable parameters. In each round, $s = 10,000$ users are randomly selected. These selected users locally compute the gradient vector using Stochastic Gradient Descent (SGD), then subsample 0.15\% of the gradient vector dimensions (with unsampled dimensions set to zero), and clip the gradient values at a threshold of $c = 0.00015$. Each user submits a sanitized version of the subsampled, clipped gradient vector as their report.

\begin{table}[]
\caption{Model architecture of the neural network.}
\label{tab:cnn}
\centering
\fontsize{8}{9}\selectfont
\begin{tabular}{|c|c|}
\hline
\textbf{Layer} & \textbf{Parameters} \\ \hline
Convolution & 8 filters of 4$\times$4, stride 2 \\ \hline
Max-pooling & 2$\times$2 \\ \hline
Convolution & 6 filters of 5$\times$5, stride 2 \\ \hline
Max-pooling & 2$\times$2 \\ \hline
Softmax & 10 units \\ \hline
\end{tabular}
\vspace*{-0.5em}
\end{table}

\noindent\textbf{LDP randomizers} We use all LDP randomizers as before except for PlanarLaplace, which is unsuitable for multi-dimensional gradient vectors. The randomizers sanitize the non-zero values in the subsampled gradient vectors while leaving the zero values unaffected.

\noindent \textbf{Server-side algorithm} The server performs the following steps: first, it takes a sanitized gradient vector and uses randomizer-specific algorithms to estimate the true values of the sampled dimensions. These vectors with estimated values are denoted as $\hat{g}_i$. The server then aggregates the $\hat{g}_i$ vectors into a global gradient vector using a simple sum-and-average method. This global gradient vector is published so that users can update their local models. Additionally, the server computes the Shapley value for each user, which measures the marginal contribution of each $\hat{g}_i$ (see Appendix \ref{app:expFL} for details\editf{ in the full version \cite{wang2025PIC}}).

\noindent \textbf{Experimental results} We train each model for 80 rounds using Federated Learning. The utility comparison results are summarized in Tables \ref{tab:flincentive1} and \ref{tab:flincentive3}, each reflecting a different global privacy budget. In these tables, we present the single report utility (Gradient $\ell_2$ error) and the task-specific utility (Shapley $\ell_2$ error). The Gradient $\ell_2$ error is calculated using the estimated gradients $\hat{g}_i$ and the true gradients $g_i$ (clipped and with all unsampled dimensions set to 0): $\frac{ \sum_{i\in S}\|\hat{g}_i-{g}_i\|_2}{|S|}$, where $S$ is the set of randomly selected users. The Shapley $\ell_2$ error is computed as follows:$
\frac{\sum\nolimits_{i\in S}\|\widehat{\texttt{Shapley}}_i - {\texttt{Shapley}}_i \|_2}{|S|}
$, where \(\widehat{\texttt{Shapley}}_i\) is derived from \(\hat{g}_i\) and \({\texttt{Shapley}}_i\) is obtained from the unsanitized \(g_i\). For reference, we also include the final models' accuracy in the tables.

\begin{table}[hbt]
\caption{Utility comparison of federated learning with incentives: global privacy budget $(1, 0.01/50000)$.}
\label{tab:flincentive1}
\centering
\fontsize{8}{9}\selectfont
\begin{tabular}{|c|c|c|c|c|}
\hline
\textbf{Setting} & \textbf{\begin{tabular}[c]{@{}c@{}}Randomization\\ Mechanism\end{tabular}} & \textbf{\begin{tabular}[c]{@{}c@{}}Test\\ Accuracy\end{tabular}} & \textbf{\begin{tabular}[c]{@{}c@{}}Gradient\\ $\ell_2$ Error\end{tabular}} & \textbf{\begin{tabular}[c]{@{}c@{}}Shapley\\ $\ell_2$ Error 
  \end{tabular}} \\ \hline
\multirow{5}{*}{\begin{tabular}[c]{@{}c@{}}local\\ model\end{tabular}} 
 & Staircase \cite{geng2015staircase} & $18.29\%$ & $15.30$ & $0.0586$  \\ \cline{2-5} 
 & Squarewave \cite{li2020estimating} & $13.91\%$  & $13.23$ & $0.0590$ \\ \cline{2-5} 
 & PrivHS \cite{duchi2018minimax} & $28.17\%$ & $2.44$  &  $0.0587$ \\ \cline{2-5}
 & Laplace \cite{dwork2008differential} & $21.22\%$ & $2.53$  &  $0.0575$ \\ \cline{2-5}
 & PrinUnit \cite{bhowmick2018protection} & $23.28\%$ & $2.37$  &  $0.0592$ \\ \cline{2-5}
 & Minkowski & $18.10\%$ & $14.45$  &  $0.0571$ \\ \hline
\multirow{5}{*}{\begin{tabular}[c]{@{}c@{}}PIC\\ model\end{tabular}}
 & Staircase \cite{geng2015staircase} & $32.51\%$ & $0.781$ & $0.0563$  \\ \cline{2-5} 
 & Squarewave \cite{li2020estimating} & $33.75\%$  & $0.604$ & $0.0585$ \\ \cline{2-5} 
 & PrivHS \cite{duchi2018minimax} & $65.47\%$ & $0.267$  &  $0.0428$ \\ \cline{2-5}
 & Laplace \cite{dwork2008differential} & $72.6\%$ & $0.154$  &  $0.0508$ \\ \cline{2-5}
 & PrinUnit \cite{bhowmick2018protection} & $74.02\%$ & $0.157$  &  $0.0438$ \\ \cline{2-5}
 & \textbf{Minkowski} & $\mathbf{78.68\%}$ & $\mathbf{0.093}$  &  $\mathbf{0.0363}$ \\ \hline
\end{tabular}
\end{table}
}

\begin{table}[hbt]
\caption{Utility comparison of federated learning with incentives: global privacy budget $(3, 0.01/50000)$.}
\label{tab:flincentive3}
\centering
\fontsize{8}{9}\selectfont
\begin{tabular}{|c|c|c|c|c|}
\hline
\textbf{Setting} & \textbf{\begin{tabular}[c]{@{}c@{}}Randomization\\ Mechanism\end{tabular}} & \textbf{\begin{tabular}[c]{@{}c@{}}Test\\ Accuracy\end{tabular}} & \textbf{\begin{tabular}[c]{@{}c@{}}Gradient\\ $\ell_2$ Error\end{tabular}} & \textbf{\begin{tabular}[c]{@{}c@{}}Shapley\\ $\ell_2$ Error\end{tabular}} \\ \hline
\multirow{5}{*}{\begin{tabular}[c]{@{}c@{}}local\\ model\end{tabular}} 
 & Staircase \cite{geng2015staircase} & $19.68\%$ & $5.09$ & $0.0519$  \\ \cline{2-5} 
 & Squarewave \cite{li2020estimating} & $18.24\%$  & $4.34$ & $0.0574$ \\ \cline{2-5} 
 & PrivHS \cite{duchi2018minimax} & $41.52\%$ & $0.850$  &  $0.0561$ \\ \cline{2-5}
 & Laplace \cite{dwork2008differential} & $28.64\%$ & $0.845$  &  $0.0575$ \\ \cline{2-5}
 & PrinUnit \cite{bhowmick2018protection} & $39.02\%$ & $0.803$  &  $0.0537$ \\ \cline{2-5}
 & Minkowski & $20.73\%$ & $3.731$  &  $0.0545$ \\ \hline
\multirow{5}{*}{\begin{tabular}[c]{@{}c@{}}PIC\\ model\end{tabular}}
 & Staircase \cite{geng2015staircase} & $44.58\%$ & $0.545$ & $0.0563$  \\ \cline{2-5} 
 & Squarewave \cite{li2020estimating} & $53.02\%$  & $0.407$ & $0.0552$ \\ \cline{2-5} 
 & PrivHS \cite{duchi2018minimax} & $74.49\%$ & $0.190$  &  $0.0383$ \\ \cline{2-5}
 & Laplace \cite{dwork2008differential} & $74.86\%$ & $0.104$  &  $0.0335$ \\ \cline{2-5}
 & PrinUnit \cite{bhowmick2018protection} & $77.42\%$ & $0.098$  &  $0.0253$ \\ \cline{2-5}
 & \textbf{Minkowski}  & $\mathbf{83.43\%}$ & $\mathbf{0.055}$  &  $\mathbf{0.0219}$ \\ \hline
\end{tabular}
\end{table}

We observe in the tables that the utility in the LDP model is significantly worse than in the PIC model, as expected. This is consistent across all metrics and the final model accuracy. In the local model, the performance of Staircase, Squarewave, and Minkowski mechanisms is poorer compared to the others, but for different reasons. Staircase and Squarewave are designed for single-dimensional data, requiring the privacy budget to be split across dimensions when sanitizing vectors, which leads to reduced utility. For Minkowski, the issue lies in the 80-round federated learning process, where the privacy budget for each round is relatively small (e.g., between 0.2 and 0.75). As previously mentioned, the Minkowski is intended to achieve better utility with a large local privacy budget. In small budget scenarios, it offers no advantage and may even perform worse. In the PIC model, the performance of Staircase and Squarewave remains poor. However, Minkowski now performs the best among all randomizers. This improvement is due to the privacy amplification in PIC, which increases the local budget for each round to between 4.0 and 6.0.

\section{Discussions}\label{sec:dis}
\added{
\textbf{Decentralized download-then-compute settings.} While our protocol initially allows a server to perform the computation $f$, it is also possible for users to first download $L$ and then compute $f$ locally. This download-then-compute paradigm can be seen as a special decentralized case of our PIC model (i.e. the server simply computes $f(x)=x$). It would also rely on the one-time random key to perform post-computation communication. The PIC protocol offers several advantages over the download-then-compute paradigm: (1) PIC accommodates more tasks, e.g. those dependent on private server information (validation dataset in federated learning with incentives). (2) Download-then-compute requires public 
$f$, while PIC allows proprietary server-side algorithms. (3) Better privacy. E.g. in taxi-hailing, download-then-compute allows each user to compute all matching pairs and location distribution, while PIC publishes encrypted matching results (without leaking unnecessary information).
}

\added{\textbf{Necessity of grouping users.} While many applications involve a single type of participant (e.g., social systems and federated learning), other tasks require interaction between multiple participant types, where the server must use the type information of participants to execute the algorithm $f$ (e.g., in taxi-hailing services with types of drivers and users). From a privacy-preservation standpoint, it is also possible to apply varying levels of privacy protection across different groups.} 



\added{\textbf{Worst-case leakage and privacy under active attacks.} We assume that an attacker can corrupt and control a subset of parties, but the data of the uncorrupted users remains $(\epsilon_c, \delta)$-private against the attacker. The specific value of $\epsilon_c$ depends on the number of uncorrupted users within the same group (see Theorem \ref{the:privacy}, Eq. \ref{eq:cloneamplify}). In the worst-case scenario, where all but one user are corrupted, the privacy guarantee is reduced to $\epsilon$-LDP, as provided by the local randomizer, i.e., there is no privacy amplification through shuffling.
}


\newadded{\textbf{Multiple executions and reuse keys.} When multiple/sequential PIC executions are needed (e.g., in federated learning), a fresh one-time user key for each execution is required, ensuring the adversary cannot correlate user activities. According to Eq. \ref{eq:cloneamplify} and sequential composition of DP \cite{dwork2008differential}, the overall privacy consumption for $k$ executions is $\tilde{O}(\sqrt{k \cdot e^{\epsilon} / {n_i^*}})$. Otherwise, reusing keys in $k$ executions will consume more privacy (to a level of shuffling $n_i^*$ messages each is $k \cdot \epsilon$-LDP, i.e., $\tilde{O}(\sqrt{e^{k \cdot \epsilon} / {n_i^*}})$).}


\added{
\textbf{Multi-message settings.} The PIC model allows each client to send one single message, to align with most non-private applications. However, restricting each client to sending one message in the shuffle model has intrinsic privacy amplification and utility limitations (as seen in statistical analyses within the shuffle model \cite{ghazi2021power}). Extending PIC model to multi-message settings is an interesting future direction, though it is quite challenging (e.g., in taxi-hailing services, the messages from a driver must appear independent, then each message may match a different customer, leading to an oversale).
}

\section{Conclusion}\label{sec:conclusion}
Privacy-preserving computation with differential privacy holds great potential for leveraging personal information. While the shuffle model offers a rigorous DP guarantee with enhanced utility, its application is confined to statistical tasks. In this paper, we introduce a novel paradigm called Private Individual Computation (PIC), which extends the shuffle model to scenarios where each user requires personalized outputs from the computation. We demonstrate that PIC can be realized using an efficient protocol that relies on minimal cryptographic operations while maintaining the advantages of privacy amplification through shuffling. To further enhance utility, we developed a local randomizer specifically designed for PIC. We provide formal proofs of the protocol's security and privacy, as well as the asymptotic optimality of the randomizer. Extensive experiments validate the superiority of the PIC protocol and the randomizer, showcasing their performance across three major application scenarios and various real-world datasets.

\noindent\deleted{\textbf{Future directions.} The relaxed notion of local metric DP, or geo-indistinguishability \mbox{\cite{andres2013geo}}, also shows shuffle privacy amplification effects \mbox{\cite{wang2023unified}}, making it a promising area for future research to further enhance utility in the PIC model.}

\section*{Acknowledgments}

We thank the anonymous reviewers for their insightful suggestions and comments. This work is supported by National Key Research and Development (R\&D) Program (Young Scientist Scheme No.2022YFB3102400), National Natural Science Foundation of China (No.62372120, No.62302118, No.62372125, No.62261160651, No.62102108), the Guangdong Natural Science Funds for Distinguished Young Scholar under Grant 2023B1515020041, Natural Science Foundation of Guangdong Province of China (No.2022A1515010061), and Guangzhou Basic and Applied Basic Research Foundation (No.2025A03J3182). Di Wang is supported in part by the funding BAS/1/1689-01-01, URF/1/4663-01-01,  REI/1/5232-01-01,  REI/1/5332-01-01,  and URF/1/5508-01-01  from KAUST, and funding from KAUST - Center of Excellence for Generative AI, under award number 5940. 




\section*{Ethics considerations}
This work does not present any ethical issues. In fact, it directly mitigate privacy concerns associated with computational tasks involving sensitive user data, such as location data, preference data, and gradient information. All experiments are conducted using public datasets, ensuring no exposure of personal information. Moreover, we strictly adhere to ethical guidelines throughout the research process, ensuring that no other ethical concerns arise.

\section*{Open science}
In alignment with the principles of open science, our code and experimental data (for the three exemplar applications of the PIC model in Section \ref{sec:exp}) are publicly accessible at \url{https://zenodo.org/records/14710367} (please refer to README.md for more details). 


\bibliographystyle{plain}
\bibliography{refs}




\balance
\appendix

\section{Error Upper Bounds of Minkowski Response Mechanism}\label{app:errorupper}
We now study the error bound of the Minkowski response in the PIC model. We start with analyze the mean squared error formula of the mechanism given fixed local budget $\epsilon$ and cap area radius parameter $r$. Then, we apply the global privacy budget $(\epsilon_c,\delta)$ and the privacy amplification bound in Theorem \ref{the:privacy}, to deduce a feasible local budget $\epsilon$ and optimized radius parameter afterward. To deal with both $\ell_2$-norm and $\ell_{+\infty}$-norm bounded domain, we introduce a more general notation $\mathbb{B}_{p,r}(x)$ to represent the $\ell_2$-norm hyperball with radius $r$ centered at any $x\in \mathbb{R}^d$:
$$\mathbb{B}_{p,r}(x)=\{x'\ |\ x'\in \mathbb{R}^d\ and\ \|x'-x\|_p\leq r\},$$ and a general notation  $\dom{Y}_{p,q,r}$ to present the following $\ell_q$ expanded domain:
\[\dom{Y}_{p,q,r}=\{x\ |\  x\in \mathbb{R}^d\ and\ \exists x'\in \mathbb{B}_{p,1}\ that\ x\in \mathbb{B}_{q,r}(x')\}.\]


\textbf{For hyper-ball domain $\mathbb{B}_{2,1}$. } When the domain is $\ell_2$-norm bounded hyperball $\mathbb{B}_{2,1}$, we use $q=2$ for the cap area as well. In this context, we let $\beta=\frac{r^d(e^\epsilon-1)}{(1+r)^d+r^d(e^\epsilon-1)}$ and obtain the MSE bound given fixed local budget $\epsilon$ and radius $r$ as follows:
{\small
\begin{alignat*}{4}
&&&\max_{x\in \mathbb{B}_{p,1}} \mathbb{E}[\|\tilde{x}-x\|_2^2]=\max_{x\in \mathbb{B}_{p,1}}  \frac{1}{\beta^2}\cdot \textsf{Var}[y|x]\\
&&=&\max_{x\in \mathbb{B}_{p,1}} \frac{1}{\beta^2}(\beta\cdot  \mathbb{E}[\|\mathbb{B}_{q,r}(x)\|_2^2]+(1-\beta)\cdot\mathbb{E}[\|\dom{Y}_{p,q,r}(x)\|_2^2]-\beta^2\|x\|_2^2)\\
&&=&\max_{x\in \mathbb{B}_{p,1}} \frac{1}{\beta^2}(\beta (\|x\|_2^2+r^2)+(1-\beta)(1+r)^2-\beta^2\|x\|_2^2)\\
&&\leq&\frac{1}{\beta^2}(\beta (1+r^2)+(1-\beta)(1+r)^2-\beta^2)\\
&&\leq&\frac{1}{\beta^2}(\beta r^2+(1-\beta)(1+(1+r)^2))
\end{alignat*}
\normalsize
}
where $\mathbb{E}[\|\mathbb{B}\|_2^2]$ denote the expected squared distance between a (uniform-distributed) space $\mathbb{B}\subseteq \mathbb{R}^d$ and the origin point $\{0\}^d$. If the local privacy budget $\epsilon$ is relatively large (e.g., $\epsilon\geq\log((c+1)^{\frac{d+2}{2}})$ for some constant $c\geq 1$), and we specify $r=(e^{\epsilon}-1)^{2/(d+2)}-1)^{-1}$, we then have $r\leq 1/c$, $\beta\in [1/2,1]$ and:
\begin{alignat}{2}
&\max_{x\in \mathbb{B}_{p,1}} \mathbb{E}[\|\tilde{x}-x\|_2^2]\nonumber\\
\leq&4(r^2+5\frac{(1+1/r)^d}{(e^{\epsilon}-1)+(1+1/r)^d})\nonumber\\
\leq&4(r^2+5\frac{(1+1/r)^d}{e^{\epsilon}-1})\nonumber\\
\leq&4(e^{\epsilon}-1)^{-2/(d+2)}+5\frac{(e^{\epsilon}-1)^{-2/(d+2)}}{e^{\epsilon}-1})\nonumber\\
\leq&24(e^{\epsilon}-1)^{-2/(d+2)}. \label{eq:ldpmse}
\end{alignat}


Consider the case $n\geq \max\big\{16 \log(1/\delta), \frac{256(1+c)^{d+2}\log(1/\delta)}{(e^{\epsilon_c}-1)^2}\big\}$ holds for some constant $c\geq 1$, we specify local budget $\epsilon$ such that $e^{\epsilon}=\frac{n (e^{\epsilon_c}-1)^2}{256\log(1/\delta)}$ holds according to Theorem \ref{the:privacy}, and specify the radius $r$ to:
$$\Big(\big(\frac{n (e^{\epsilon_c}-1)^2}{256\log(1/\delta)}\big)^{1/(d+2)}-1\Big)^{-1}.$$
Observe that in this setting, we have $\beta\in [1/2,1]$ and $r\leq 1/c$. Then, the MSE is upper bounded as:
{\small
\begin{alignat*}{4}
&&&\max_{x\in \mathbb{B}_{p,1}} \mathbb{E}[\|\tilde{x}-x\|_2^2]\\
&&\leq&\frac{1}{\beta^2}(\beta r^2+(1-\beta)(1+(1+r)^2))\\
&&\leq&4(r^2+5\frac{(1+1/r)^d}{(e^{\epsilon}-1)+(1+1/r)^d})\\
&&\leq&4(r^2+5\frac{(1+1/r)^d}{e^{\epsilon}})\\
&&\leq&4\Big(\big(\frac{256\log(1/\delta)}{n(e^{\epsilon_c}-1)^2}\big)^{\frac{2}{d+2}}\cdot \frac{(c+1)^2}{c^2}+\frac{5((e^{\epsilon_c}-1)^2 n/(256\log(1/\delta)))^{\frac{d}{d+2}}}{(e^{\epsilon_c}-1)^2 n/(256\log(1/\delta))}\Big)\\
&&\leq&36\Big(\frac{256\log(1/\delta)}{n(e^{\epsilon_c}-1)^2}\Big)^{\frac{2}{d+2}}.
\end{alignat*}
\normalsize
}
Therefore, we establish Theorem \ref{the:errorupper}. With sufficiently large size $n$ of the amplification population, the derived error bound matches the lower bound in Theorem \ref{the:errorlower}. 

\textbf{For hyper-cube domain $\mathbb{B}_{\infty,1}$. } Another data domain that is commonly encountered in practical settings is the $\ell_{+\infty}$-norm bounded hypercube. We use $q=+\infty$ as well for the cap area, then we have volumes $V(\mathbb{B}_{2,r})=(2r)^d$, $V(\dom{Y}_{\infty,\infty,r})=(2+2r)^d$, and let $\beta=\frac{r^d(e^\epsilon-1)}{(1+r)^d+r^d(e^\epsilon-1)}$. For fixed local budget $\epsilon$ and radius $r$, the mean squared error bound is:
\small
\begin{alignat*}{4}
&&&\max_{x\in \mathbb{B}_{p,1}} \mathbb{E}[\|\tilde{x}-x\|_2^2]    = \frac{1}{\beta^2}\cdot \text{Var}[y]\\
&&\leq& \frac{d}{\beta^2}(\beta (1+r^2/3)+(1-\beta)(1+r)^2/3)-d\\
&&\leq& \frac{d}{3\beta^2}(\beta r^2+(1-\beta)((1+r)^2+3(1-\beta)))
\end{alignat*}
\normalsize

Consider the case $n\geq \max\{16 \log(1/\delta), \frac{32(1+c)^{d+2}\log(1/\delta)}{(e^{\epsilon_c}-1)^2}\}$ holds for some constant $c\geq 1$, we specify local budget $\epsilon$ such that $e^\epsilon=\frac{n (e^{\epsilon_c}-1)^2}{256\log(1/\delta)}$ holds according to Equation \ref{eq:cloneamplify}, and specify the radius $r$ as: $$\Big(\big(\frac{n (e^{\epsilon_c}-1)^2}{256\log(1/\delta)}\big)^{1/(d+2)}-1\Big)^{-1}.$$
In this setting, we have $\beta\in [1/2,1]$ and $r\leq 1/c$, we thus obtain:
\small
\begin{alignat*}{4}
&&&\max_{x\in \mathbb{B}_{p,1}} \mathbb{E}[\|\tilde{x}-x\|_2^2]\\
&&\leq& \frac{d}{3\beta^2}(r^2+7(1-\beta))\\
&&\leq& \frac{d}{3\beta^2}(r^2+7\frac{(1+1/r)^d}{(e^{\epsilon}-1)+(1+1/r)^d})\\
&&\leq& \frac{d}{3\beta^2}(r^2+7\frac{(1+1/r)^d}{e^{\epsilon}})\\
&&\leq&\frac{4d}{3}\Big(\big(\frac{256\log(1/\delta)}{n(e^{\epsilon_c}-1)^2}\big)^{\frac{2}{d+2}}\cdot \frac{(c+1)^2}{c^2}+7\big(\frac{256\log(1/\delta)}{n(e^{\epsilon_c}-1)^2}\big)^{\frac{2}{d+2}}\Big)\\
&&\leq&{15 d}\cdot \Big(\frac{256\log(1/\delta)}{n(e^{\epsilon_c}-1)^2}\Big)^{\frac{2}{d+2}}.
\end{alignat*}
\normalsize

Alternatively, one may firstly transform the $\ell_{+\infty}$-norm vector into a $\ell_{2}$-norm bounded one, and utilizing the mechanism for hyper-ball. Similar utility can be guaranteed for both ways.

\section{Details on Experimental Implementation}\label{app:mechdetail}
In the experiments related to both spatial crowdsourcing and location-based social systems, user location data is confined within a two-dimensional cube domain $[-1,1]\times [-1,1]$. As a result:

\begin{itemize}
    \item For the Laplace mechanism, the privacy sensitivity parameter related to replacement is defined as $\Delta=4$. 
    \item In the PlanarLaplace mechanism, given that the maximum $\ell_2$-distance is $2\sqrt{2}$, we set the geo-indistinguishability parameter to $\epsilon/(2\sqrt{2})$ to ensure a fair comparison. 
    \item In mechanisms like Staircase and Squarewave, which originally operate in one-dimensional domain, the local budget is evenly distributed across two dimensions. This is crucial for generating meaningful location reports pertinent to these tasks. 
    \item For the PrivUnit mechanism, which uses an $\ell_2$-bounded unit vector as input, we convert the two-dimensional cube domain into a three-dimensional hyper-ball domain. After randomization, it's reverted back to its original two-dimensional form. To enhance performance, we further engage in a numerical search for the optimal hyper-parameter, following the approach in \cite{feldman2021lossless}.
    \item In the Minkowski response mechanism, we set $q=+\infty$ for the cap area to align with the input domain, and engage in a numerical search for the best-suited cap area radius $r$.
    \item We additionally introduce a classical mechanism by Duchi \textit{et al.} \cite{duchi2018minimax}, denoted as PrivHS, for comparison. 
\end{itemize}

In the experiments of federated learning with incentives, the (randomly) subsampled gradient vector has $6$ dimensions, thus the two-dimensional PlanarLaplace mechanism is unapplicable. For mechanisms like Staircase and Squarewave,  the local budget is evenly distributed across $6$ dimensions.

When privacy amplification via shuffling is applied in the PIC model, the parameter $\delta$ is fixed to $0.01/n_i$, where $n_i$ is the number of users in the same group $i$. In the spatial crowdsourcing applications, since one user is associated with at most one worker, the effective number $n'_{i}$ of amplification population is set to $n_i-1$; in the location-based social system applications, the effective number $n'$ of amplification population is set to $0.90\cdot n$ when neighboring radius $\tau=0.2$ and is set to $0.98\cdot n$ when neighboring radius $\tau=0.1$. 

We evaluate both the client-side and server-side running time of our protocol on a laptop computer embedded with Intel i$5$-8250U CPU @1.6GHz and $8$GB memory.

\clearpage
\nobalance

\section{Spatial Crowdsourcing}\label{app:expSC}

\begin{table*}[t]
\caption{Mean $\ell_2$-error (and running time) comparison of local $\epsilon$-LDP randomizers on location domain $[-1,1]\times[-1,1]$. All results are the expected value of $1000$ repeated experiments.}
\label{tab:errortime}
\small
\centering
\begin{tabular}{|c|l|l|l|l|l|l|l|}
\hline
\multicolumn{1}{|c|}{\textbf{randomizer}} & \multicolumn{1}{c|}{\textbf{$\epsilon=0.5$}} & \multicolumn{1}{c|}{\textbf{$\epsilon=1.0$}} & \multicolumn{1}{c|}{\textbf{$\epsilon=2.0$}} & \multicolumn{1}{c|}{\textbf{$\epsilon=3.0$}} & \multicolumn{1}{c|}{\textbf{$\epsilon=5.0$}}                         & \multicolumn{1}{c|}{\textbf{$\epsilon=8.0$}}                         & \multicolumn{1}{c|}{\textbf{$\epsilon=10.0$}}                        \\ \hline
PrivUnit \cite{bhowmick2018protection}            & {8.94} (49us)                                    & 4.68 (42us)                & 2.25 (61us)               & 1.44 (63us)               & 0.81 (113us)  & 0.32 (305us)  & 0.18 (868us)  \\ \hline
PrivUnitG\cite{asi2022optimal}           & \textbf{8.73} (76ms)                                    & 4.63 (75ms)               & 2.27 (77ms)               & 1.51 (76ms)               & 0.96 (74ms)                                       & 0.63 (75ms)                                       & 0.53 (81ms)                                                            \\ \hline
Laplace \cite{dwork2008differential}             & 12.97 (0.1us)                                  & 6.56 (0.1us)              & 3.27 (0.1us)              & 2.13 (0.1us)              & 1.30 (0.1us)                                      & 0.81 (0.1us)                                      & 0.64 (0.1us)                                                           \\ \hline
PlanarLaplace \cite{andres2013geo}       & 11.17 (12us)                                   & 5.63 (11us)               & 2.84 (11us)               & 1.88 (11us)               & 1.14 (11us)                                       & 0.71 (12us)                                       & 0.56 (12us)                                                            \\ \hline
Staircase\cite{geng2015staircase}           & 13.19 (49us)                                   & 6.40 (48us)               & 3.13 (46us)               & 2.01 (48us)               & 1.05 (44us)                                       & 0.46 (51us)                                       & 0.28 (50us)                                                            \\ \hline
SquareWave\cite{li2020estimating}          & 11.87 (1.9us)                                  & 5.72 (2.5us)              & 2.65 (1.9us)              & 1.68 (2.2us)              & 0.92 (1.7us)                                      & 0.53 (2.2us)                                      & 0.42 (2.1us)                                                           \\ \hline
MinkowskiResponse           & {10.42} (0.7us)                         & \textbf{4.50} (0.6us)     & \textbf{1.78} (0.8us)     & \textbf{0.98} (0.8us)     & \textbf{0.39} (0.6us)                             & \textbf{0.14} (0.7us)                             & \textbf{0.074} (0.8us)                                                 \\ \hline
\end{tabular}
\end{table*}

\begin{table*}[]
\caption{Running time and communication overheads of spatial crowdsourcing in the PIC model.}
\label{tab:timecomm}
\small
\centering
\begin{tabular}{llcllc}
\hline
\multicolumn{1}{|l|}{\textbf{User-side Procedures}} & \multicolumn{1}{l|}{\textbf{Time}} & \multicolumn{1}{c|||}{\textbf{Communication}} & \multicolumn{1}{l|}{\textbf{Server-side Procedures}} & \multicolumn{1}{l|}{\textbf{Time}} & \multicolumn{1}{c|}{\textbf{Communication}} \\ \hline
\hline
\multicolumn{1}{|l|}{location randomization} & \multicolumn{1}{l|}{0.7us} & \multicolumn{1}{c|||}{-} & \multicolumn{1}{l|}{decrypt one message} & \multicolumn{1}{l|}{2.1ms} & \multicolumn{1}{c|}{-} \\
\multicolumn{1}{|l|}{encrypt information} & \multicolumn{1}{l|}{2.9ms} & \multicolumn{1}{c|||}{-} & \multicolumn{1}{l|}{min-weight match (GMission)} & \multicolumn{1}{l|}{31ms} & \multicolumn{1}{c|}{1.2MB} \\
\multicolumn{1}{|l|}{send message to shuffler} & \multicolumn{1}{c|}{-} & \multicolumn{1}{c|||}{1.3KB} & \multicolumn{1}{l|}{maximum match (GMission)} & \multicolumn{1}{l|}{15ms} & \multicolumn{1}{c|}{1.2MB} \\
\multicolumn{1}{|l|}{retrieve matching result} & \multicolumn{1}{c|}{-} & \multicolumn{1}{c|||}{1.8KB} & \multicolumn{1}{l|}{min-weight match (EverySender)} & \multicolumn{1}{l|}{79ms} & \multicolumn{1}{c|}{4.3MB} \\
\multicolumn{1}{|l|}{key agreement with worker} & \multicolumn{1}{l|}{2.6ms} & \multicolumn{1}{c|||}{-} & \multicolumn{1}{l|}{maximum match (EverySender)} & \multicolumn{1}{l|}{63ms} & \multicolumn{1}{c|}{4.3MB}\\ \hline
\end{tabular}
\end{table*}

\textbf{Dataset descriptions. } The GMission dataset originates from a spatial crowdsourcing platform for scientific simulations. It contains information about every task, including its description, location, time of assignment, and deadline (in minutes). Furthermore, it provides data about each worker, comprising their location, maximum activity range (in kilometers), etc. EverySender, on the other hand, represents a campus-based spatial crowdsourcing platform, facilitating everyone to post micro tasks like package collection or to act as a worker. Similar to GMission, EverySender dataset also carries detailed information for every task and worker. We assume each worker's capacity as one and, for simplicity, we consider only the location information of the taskers/workers for matching purposes.

Before exploring spatial crowdsourcing tasks, we first evaluate the utility and efficiency of our proposed Minkowski response mechanism against existing mechanisms in the LDP model, as presented in Table \ref{tab:errortime}. This evaluation focuses on errors and running times associated with reporting locations under local $\epsilon$-DP constraints. We quantify the error using the expected $\ell_2$ distance between each true and noisy location pair. Notably, the Minkowski response showcases an $\ell_2$ error comparable to state-of-the-art (SOTA) mechanisms when $\epsilon\leq 1.0$. Its superiority is evident when $\epsilon\geq 2.0$, displaying a significantly reduced error. This supports the theoretical assertions made in Section \ref{sec:ldpmechanism}, underscoring the importance of the Minkowski response for the shuffle model, over other existing randomizers. Additionally, the Minkowski response is highly efficient, demanding less than $1$us on the user end. This speed is nearly equivalent to merely adding Laplace noises and is roughly $50$x faster than the previously SOTA PrivUnit.

We present the computation and communication overheads of the PIC model for spatial crowdsourcing in Table \ref{tab:timecomm}. The total running time for each user is only a few milliseconds (with the data randomization time being negligible), and the total communication overhead is under $4$KB. On the server side, the running time of the matching algorithm on noisy plaintext is consistent with non-private settings, scaling with the number of users/workers and taking dozens of milliseconds. The additional decryption/encryption runtime is capped at several seconds, mirroring the costs of HTTPS decryption/encryption. Specifically, it utilizes ECDHE for key agreement and AES encryption as with HTTPS on TLS 1.3 \cite{rescorla2018transport} for client-server communication. In essence, the running times on both the client and server sides of the PIC model closely resemble those in non-private settings that use HTTPS communication (for a fair comparison).

\section{Location-based Social Systems}
\textbf{Dataset descriptions. } Gowalla is a location-based social networking website where users share their locations by checking-in. The Gowalla consists a total of 6,442,890 check-ins of $138368$ users over the period of Feb. 2009 - Oct. 2010, in the San Francisco, CA. The Foursquare dataset contains check-ins in New York city collected for about 10 month (from 12 April 2012 to 16 February 2013). It contains 227,428 check-ins in New York city. Each check-in is associated with its time stamp, its GPS coordinates and its semantic meaning (represented by fine-grained venue-categories). 


To further illustrate the effects of neighboring radius $\tau$ on performances, we consider $(\tau=0.1)$-radius nearest neighbor queries, where each user has several tens or hundreds of neighbors. The experimental results on the LDP model are presented in Figure \ref{fig:LBSf1local01}, and the results on the PIC shuffle model are presented in Figures \ref{fig:LBSdistanceshuffle01} and \ref{fig:LBSf1shuffle01}. Considering post-computation user-to-user communications, and each user has about $n\cdot\frac{\pi\cdot 0.1^2}{2^2}\approx n\cdot 0.8\%$ neighbors, we use $\lfloor n\cdot 98\%\rfloor$ when deriving the local budget given the global budget $\epsilon_c$. The combination of the PIC model and the Minkowski response mechanism outperforms competitors in most scenarios. The Minkowski randomizer can achieve an $F_1$ score of 0.75 with a stringent global budget $\epsilon_c=1$ and over 0.80 with $\epsilon_c=3$. Compared to the $\tau=0.2$ scenarios, the $\tau=0.1$ cases retrieve nearer (and fewer) neighbors, but can be less stable under DP noises (i.e., has lower F1 scores).

\begin{figure}[!htb]
\begin{center}
\centerline{\includegraphics[width=78mm]{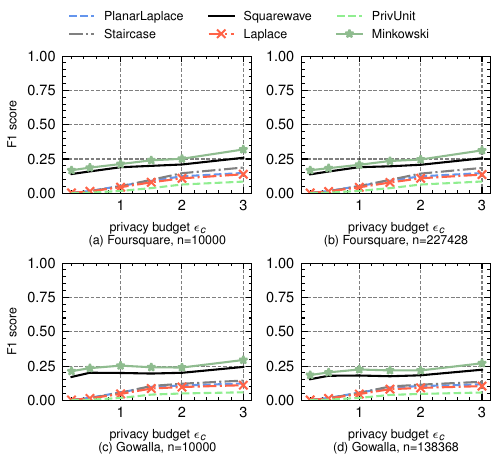}}
\vspace*{-0.8em}
\caption{F1 scores of nearest neighbor queries (LDP model) with radius $\tau=0.1$.}
\label{fig:LBSf1local01}
\end{center}
\vspace*{-1.2em}
\end{figure}

\begin{figure}[!htb]
\begin{center}
\centerline{\includegraphics[width=78mm]{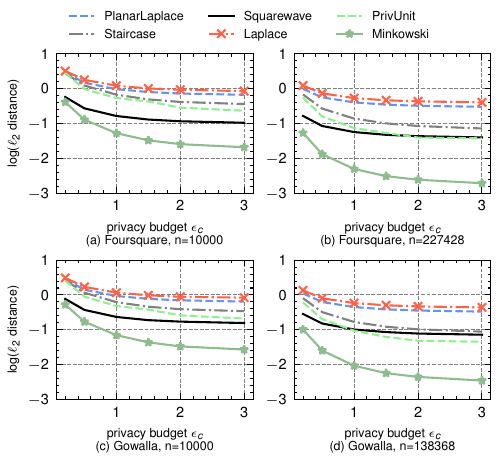}}
\vspace*{-0.8em}
\caption{Expected $\ell_2$ distances of reported locations in location-based social systems with PIC model with radius $\tau=0.1$.}
\label{fig:LBSdistanceshuffle01}
\end{center}
\vspace*{-1.2em}
\end{figure}

\begin{figure}[!htb]
\begin{center}
\centerline{\includegraphics[width=78mm]{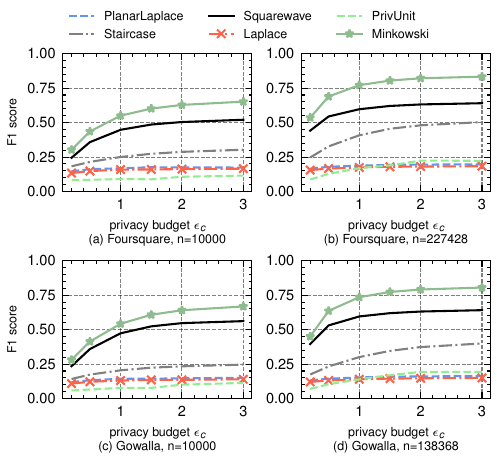}}
\vspace*{-0.8em}
\caption{F1 scores of nearest neighbor queries (PIC model) with radius $\tau=0.1$.}
\label{fig:LBSf1shuffle01}
\end{center}
\vspace*{-0.5em}
\end{figure}




\section{Federated Learning with Incentives}\label{app:expFL}
Among 60000 images in the MNIST dataset, 50000 of them are designated as training samples, 5000 of them works as the validation dataset and the remaining 5000 images works as the test dataset. During the 80 rounds of (noisy) gradient aggregation and model updating, we use Adam optimizer with learning rate 0.05 and decay rate 0.99.

\begin{table*}[!htbp]
\caption{Comparison of various approaches to multi-party PIC computation, with $C$ gates and $h$ depth of computation circuits, and $(\epsilon,\delta)$-DP where $\epsilon=O(1)$, where $\lambda$ is the security parameter, the $n$ is number of clients/inputs, $d$ is the dimension of inputs. The \emph{Privatization Error} is the expected mean squared error between an input and its reported values due to differential privacy; the \emph{Algo. on Plaintext} indicates whether the computation circuit is running on un-encrypted plaintext.}
\label{tab:mpec}
\centering
\fontsize{8}{9}\selectfont
\begin{tabular}{l|c|c|c|c|c}
\hline
\hline
\textbf{Paradigm} & \textbf{\begin{tabular}[c]{@{}c@{}}Client \\Comm. Rounds\end{tabular}} & \textbf{\begin{tabular}[c]{@{}c@{}}Client \\Comm. Bits\end{tabular}} & \textbf{\begin{tabular}[c]{@{}c@{}}Orchestrator\\ Comput. Costs\end{tabular}} & \textbf{\begin{tabular}[c]{@{}c@{}} Algo. on\\ Plaintext? \end{tabular}} & \textbf{\begin{tabular}[c]{@{}c@{}}Privatization\\ Error\end{tabular}} \\ \hline
plaintext & $1$ & $O(d)$ & $O(nd)$ & ${\color{green}\yesmark}$ & -  \\ \hline
plaintext with secure communication & $1$ & $O(\max\{d,\lambda\})$ & $O(n\cdot\max\{d,\lambda\})$ & ${\color{green}\yesmark}$ & -  \\ \hline
\hline
GMW/BGW protocols \cite{goldreich2019play,ben2019completeness} & $h$ & $O(Poly(n)\cdot C\cdot \lambda)$ & $O(Poly(n)\cdot C\cdot \lambda)$ & ${\color{orange}\nomark}$ & - \\ \hline
BMR garbled circuits \cite{yao1986generate,beaver1990round} & $3$ & $O(Poly(n)\cdot C\cdot \lambda)$ & $O(Poly(n)\cdot C\cdot \lambda)$ & ${\color{orange}\nomark}$ & - \\ \hline
Ishai's anonymous model \cite{ishai2006cryptography}+ABT\cite{applebaum2021perfect} & $3$ & $O(Poly(n)\cdot C\cdot \lambda)$ & $O(Poly(n)\cdot C\cdot \lambda)$ & ${\color{orange}\nomark}$ & - \\ \hline
Beimel's anonymous model \cite{beimel2020round} & $2$ & $O(Poly(n)\cdot C\cdot \lambda)$ & $O(Poly(n)\cdot C\cdot \lambda)$ & ${\color{orange}\nomark}$ & - \\ \hline
\hline
local DP model \cite{kasiviswanathan2011can} with secure communication & $1$ & $O(\max\{d,\lambda\})$ & $O(n\cdot\max\{d,\lambda\})$ & ${\color{green}\yesmark}$ & $O(d/\epsilon^2)$ \\ \hline
PIC model & $1$ & $O(\max\{d,\lambda\})$ & $O(n\cdot\max\{d,\lambda\})$ & ${\color{green}\yesmark}$ & $\tilde{O}(1/(n \epsilon^2)^{\frac{2}{d+2}})$ \\ \hline
\hline
\end{tabular}
\end{table*}

We use Shapley value to measure the contribution of each gradient report. We define the utility function of Shapley payoff as cosine similarity between aggregated private gradient and the true gradient $\texttt{grad}_{val}$ (which is the average gradient over the validation dataset):
\[U(S)=\frac{<\texttt{grad}_{val}, \sum_{i\in S} g_i>}{\|\texttt{grad}_{val}\|_2\cdot\|\sum_{i\in S} g_i\|_2},\]
so as to efficiently approximate the negative loss function over the validation dataset. We note that one can use other Shapley-value utility functions within our PIC model, we choose this one for efficiency.  
Then, the Shapley value of one single gradient update $g_i$ is computed as follows:
\[\texttt{Shapley}_i=\frac{1}{n}\sum_{k=1}^n {n\choose k-1}\sum_{\substack{S\subseteq [n]\backslash \{i\},\\|S|=k-1}} U(S\cup \{i\})-U(S).\]
In cases the server receives sanitized gradient information, the corresponding sanitized version of Shapley values is computed in the same way, except using sanitized gradient in the utility function.

The Shapley values themselves may severely leak sensitive information about user data \cite{luo2022feature}. Despite the vast non-linear computations involved in evaluating Shapley values with neural networks, no rigorous protection has been offered for federated learning with incentives in existing literature. In this study, we have demonstrated the feasibility of \emph{computing Shapley value without incurring additional privacy loss}.

\section{Comparison of Private Permutation-equivariant Multi-party Computation}\label{app:mpec}

We perform comparisons across a broader range of permutation-equivariant computation tasks. These tasks include bipartite matching and federated learning with incentives. We detail these comparisons in Table \ref{tab:mpec}, emphasizing general permutation-equivariant computation tasks characterized by $C$ arithmetic circuit gates and $h$ circuit depth. Specifically:

\begin{itemize}
\item The nearest neighboring task aims to find $k$ nearest neighbors for each party, the $C$ is of the order $n^3$, and the $h$ is of the order $\log n$.
\item In minimum weight bipartite matching using LAPJVsp algorithm \cite{jonker1988shortest}, both $C$ and $h$ are of the order $n^3$ where $n$ is the number of nodes in the bipartite graph.
\item In maximum bipartite matching using Hopcroft-Karp algorithm \cite{hopcroft1973n}, $C$ is of the order $E\sqrt{n}$, $h$ is of the order $\sqrt{n}$, where $E$ is the number of edges in the bipartite graph.
\item For federated learning with Shapley incentives, which uses accuracy on the validation dataset as the utility function, $C=n_{val}\cdot N\cdot M$ can become exceptionally large, where $N$ isthe number  of neurons in neural networks, $n_{val}$ is the number of samples in the validation dataset, and $M$ is the number of Monte Carlo evaluations. The $h$ is the depth of the neural network.
\end{itemize}


We note that the Ishai's MPC model \cite{ishai2006cryptography} and Beimel's MPC model \cite{beimel2020round} also use a role of message shuffler, but their shufflers must be fully trustable. In general, MPC-based approaches impose high computation and interaction overheads, especially when the computation algorithm gets sophisticated. As comparison, the proposed paradigm in this work permits running arbitrary algorithms on plaintext, and requires only \emph{semi-trustness/honest-but-curious} assumption on the shuffler (as the shuffler only sees ciphertexts, see Section \ref{sec:protocol}).

Our PIC model does not necessitate algorithmic computation on the user side. Additionally, it permits the orchestrating server to operate on plaintext, making it considerably more efficient in computation and communication compared to MPC-based methods. In fact, the cost of the model grounded on hybrid encryption (covered in Section \ref{sec:protocol}) virtually mirrors non-private settings that maintain secure communications, such as those employing HTTPS. Besides, our model is compatible with existing achievements and future advancements in server-side algorithms (e.g., noise-aware matching algorithms for spatial crowdsourcing and combinatorial optimization).

\section{Security Proof}\label{app:proof}
This part provides security proof for the PIC protocol in \S\ref{sec:protocol}.  

\subsection{Protocol Setting and Security Goals}\label{app:proof:threat-model}

\textbf{Protocol setting}.
Our permutation-invariant computation (PIC) protocol involves $m$ groups of clients and one single computing server $S$.
These $m$ groups of parties want to jointly compute some pre-defined computing task, formalized as a multi-input function $f$, with the help of $S$. At the end of the protocol, each client may receive a function output.

\textbf{Security definitions and goals}. 
We follow the simulation-based security model with semi-honest adversaries, who will faithfully follow the protocol specifications but try to learn more information than allowed through protocol interaction. 
Security goals are formally captured as an ideal functionality $\mathcal{F}$. 
$\mathcal{F}$ receives inputs from the parties, performs computation, and sends the computation result back to the parties. 
Roughly, the security goals include \textit{privacy} and \textit{correctness}. 
Privacy requires that the adversary can only learn information as allowed but nothing more, while correctness requires that the computation is done correctly. We note that correctness is easy to achieve for semi-honest protocols.  

{\textbf{Remark on privacy guarantee}}.   
It's known that security proof for a secure computation protocol demonstrates that no additional information about parties' inputs is revealed, \textit{except} the computation output and any allowed/inherent leakage, which is well-captured in the definition of an ideal functionality. 
However, we note the computation output (combined with captured leakage) may contain a significant amount of private information about parties' inputs. 
Our PIC protocol achieves differential privacy, which provides a trade-off between privacy and utility. 
For this part, we refer to \S\ref{sec:privacy} for a formal analysis of the DP guarantees offered by our protocols.

\subsection{Ideal Functionality}\label{app:proof:functionality}

\textbf{The ideal shuffle functionality $\mathcal{F}_{\scriptscriptstyle \textsf{Shuffle}}$}. The shuffle functionality $\mathcal{F}_{\scriptscriptstyle \textsf{Shuffle}}$ receives $n$ inputs from $n$ input providers and outputs $n$ randomly permuted outputs of the original inputs. 
Possible instantiations of $\mathcal{F}_{\scriptscriptstyle \textsf{Shuffle}}$ include trusted hardware and securely evaluating a permutation network using MPC. 

\textbf{The ideal Permutation-Invariant-Computation functionality $\mathcal{F}_{\scriptscriptstyle \textsf{PIC}}$}.
The ideal PIC functionality captures the core features of our permutation invariant computation.
It receives inputs from $m$ groups of parties, adds noise to each input, randomly shuffles inputs of each group, and performs computation over the noisy inputs. 
At the end of the protocol, $\mathcal{F}_{\scriptscriptstyle \textsf{PIC}}$ sends each party a computation result, and $\mathcal{F}_{\scriptscriptstyle \textsf{PIC}}$ additionally sends all randomized inputs and all function outputs to the server.

\subsection{Proof}\label{app:proof:the-proof}
We prove the security of our protocol in the static corruption setting, where the adversary specifies the corruption parties before running the protocol. 
Let $\mathcal{C}$ be the collection of corrupted parties and $\mathcal{C} \subset G_1 \cup G_2 \cdots G_n \cup \{S\}$, and $\mathcal{H} = G_1 \cup G_2 \cdots G_n \cup \{S\} \setminus \mathcal{C}$ be the remaining honest parties. 
The proof is simulation-based. It shows that for any PPT adversary $\mathcal{A}$ who corrupts a set of parties, there exists a PPT simulator $\mathcal{S}$ that can generate a simulated view indistinguishable from the view of real-world execution. 
Note that the simulator learns the corrupted parties' input and output (i.e., computation result and allowed leakage) for the view simulation in the semi-honest security model. 
In this proof, we assume some existing ideal functionalities (e.g., $\mathcal{F}_{\scriptscriptstyle \textsf{Shuffle}}$) in the hybrid model, and we directly use existing simulators for these functionalities when needed. 
Depending on whether $S \in \mathcal{C}$, we separately prove the security of our PIC protocol in two cases as follows: 

\textbf{{Case 1: $S \notin \mathcal{C}$}}. 
In this case, all corrupted parties are clients. We construct a simulator $\mathcal{S}$ for view simulation as follows:

\begin{enumerate}
\item The simulator $\mathcal{S}$ sets up global parameters of the system as the real protocol execution, including the specification of a public key encryption scheme $\Pi=(\textsf{Gen}, \textsf{Enc},\textsf{Dec})$, security parameter $\lambda$, the server's public key $pk_c$ from invoking $\textsf{Gen}(\lambda)$, each client groups $G_i$ $(i\in[m])$, and the data randomization mechanisms $\mech{R}_i$ for each group. $\mathcal{S}$ generates public keys for all parties.

\item For the $j$-th client $u_{i,j}$ in group $G_i$, if $u_{i,j}$ is corrupted by $\mathcal{A}$, $\mathcal{S}$ runs as the real protocol execution: It randomizes $u_{i,j}$'s input $x_{i,j}$ with mechanism $\mech{R}_i$ and obtains $x_{i,j}'=\mech{R}_i(x_{i,j})$. Then the sanitized input is concatenated with  $u_{i,j}$'s public key $pk_{i,j}$, and encrypted with the server's public key $x_{i,j}''=\textsf{Enc}_{pk_c}( pk_{i,j}||x'_{i,j})$. 
If $u_{i,j}$ is not corrupted, $\mathcal{S}$ generates a ciphertext $x_{i,j}'' = \textsf{Enc}_{pk_c}(\mathbf{0})$ from the ciphertext domain (of the public encryption scheme $\Pi$) as the simulated ciphertext from $u_{i,j}$, where $\mathbf{0}$ denotes a vector of 0s of equal length as $pk_{i,j}||x'_{i,j}$.

\item $\mathcal{S}$ invokes $\mathcal{S}_{\scriptscriptstyle \textsf{Shuffle}}$ to simulate the view involved in the shuffle protocol. 
In particular, $\mathcal{S}$ has ciphertexts $\{x_{i,j}''\}_{i \in [m],j \in [n_i]}$ as the input to $\mathcal{S}_{\scriptscriptstyle \textsf{Shuffle}}$. 
Additionally, $\mathcal{S}$ generates ciphertexts $\{\widetilde{x}_{i,j}''\}_{i \in [m],j \in [n_i]}$, where each $\widetilde{x}_{i,j}''$ is generated as follows: 
    For each corrupted party $u_{i,j} \in \mathcal{C}$, $\mathcal{S}$ learns $\pi_i(j)$ from the leakage profile $\mathcal{L}$. Then $\mathcal{S}$ generates $\widetilde{x}''_{i,\pi_i(j)}$ by randomizing ciphertext $\widetilde{x}''_{i,j}$. 
    For $\widetilde{x}_{i,j}''$ corresponding to an honest party $u_{i,j}$, $\mathcal{S}$ simply generates $\widetilde{x}_{i,j}'' = \textsf{Enc}_{pk_c}(\mathbf{0})$. 
$\mathcal{S}$ then invokes $\mathcal{S}_{\scriptscriptstyle \textsf{Shuffle}}(\{x_{i,j}''\}_{i \in [m],j \in [n_i]},$ $\{\widetilde{x}_{i,j}''\}_{i \in [m],j \in [n_i]}$, $\mathcal{L})$ to generate the view for the shuffle phase, and appends it as a part of the simulated view for $\mathcal{S}$.  


\item The last piece of simulation is to ensure consistency between the prior view and the computation result of $f$. 
To this end, $\mathcal{S}$ works as follows: 
    For each corrupted party $u_{i,j} \in \mathcal{C}$, $\mathcal{S}$ obtains $y_{i, \pi_i(j)}$ from $\mathcal{F}_{\scriptscriptstyle \textsf{PIC}}$ and $\pi_i(j)$ (from the leakage profile $\mathcal{L}$). It then generates $(pk_{i,\pi_i(j)},\textsf{Enc}_{pk_{i,\pi_i(j)}}(y_{i,\pi_i(j)})$, which is then arranged as the $j$-th encrypted output within group ${G}_i$. 
    If party $u_{i,j}$ is not corrupted, $\mathcal{S}$ generates $(pk_{i,j},c_{i,j})$, where $c_{i,j} = \textsf{Enc}(\mathbf{0})$ is ciphertext of $\mathbf{0}$ with the same length as $y_{i,j}$. 
\end{enumerate} 

Below, we show the simulated view is indistinguishable from real protocol execution via the following hybrid games. 

$\mathbf{G_1}$: This is the real protocol execution. 

$\mathbf{G_2}$: $\mathbf{G_2}$ is same as $\mathbf{G_1}$, except the following difference: 
    For all ciphertexts corresponding to the honest clients, $\mathbf{G_2}$ generates the ciphertexts as the encryption of $\mathbf{0}$ with proper length.
    Due to the IND-CPA security of the public encryption scheme $\Pi$,  $\mathbf{G_2}$ is computationally indistinguishable from $\mathbf{G_1}$. 
    
$\mathbf{G_3}$: $\mathbf{G_3}$ is same as $\mathbf{G_2}$, except the following difference: 
    $\mathbf{G_3}$ uses the leakage profile $\mathcal{L}$ to arrange the ciphertexts that should be outputted from the shuffling phase. 
    In particular, for each ciphertext corresponding to a corrupted party $j$ of group $i$ after the shuffle phase, $\mathbf{G}_3$ puts the input ciphertext to the coordinate $\pi_i(j)$ and randomizes the ciphertext. 
The view distribution is identical to $\mathbf{G}_2$. 

$\mathbf{G_4}$: $\mathbf{G_4}$ is same as $\mathbf{G_3}$, except the following difference: 
    $\mathbf{G_4}$ invokes the simulator $\mathcal{S}_{\scriptscriptstyle \textsf{Shuffle}}$ to simulate the view corresponding to the view of shuffling. 
    Since our protocol works in the hybrid model, $\mathbf{G_4}$ is identical to the view from $G_3$. 
    Also note that $\mathbf{G}_4$ works the same as the simulator $\mathcal{S}$. 

Overall, the view generated by $\mathcal{S}$ is computationally indistinguishable from the view of real protocol execution in the $\mathcal{F}_{\scriptscriptstyle \textsf{Shuffle}}$-hybrid model. 

\textbf{{Case 2: $S \in \mathcal{C}$}}. 
The server $S$ is also corrupted in this case. 
We construct a simulator $\mathcal{S}$ for view simulation as follows:

\begin{enumerate}
\item The simulator $\mathcal{S}$ sets up global parameters of the system as the server does in the real protocol execution, including the specification of a public key encryption scheme $\Pi=(\textsf{Gen}, \textsf{Enc},\textsf{Dec})$, security parameter $\lambda$, the server's public key $pk_c$ from invoking $\textsf{Gen}(\lambda)$, each client groups $G_i$ $(i\in[m])$, and the data randomization mechanisms $\mech{R}_i$ for each group. 
$\mathcal{S}$ generates public keys for all parties like the real protocol execution. 

\item \label{2:step:2} 
To simulate the view corresponding to the encrypted computation output, $\mathcal{S}$ receives $(L, f(L))$ from the output of $\mathcal{F}_{\scriptscriptstyle \textsf{PIC}}$.  
$\mathcal{S}$ pareses $f(L)$ as $(\{z_{1,j}\}_{j},\cdots,\{z_{m,j}\}_{j})$; note that $(\{z_{1,j}\}_{j},\cdots,\{z_{m,j}\}_{j})$ corresponds to $\{y_{1,\pi_1(j)}\}_{j\in [n_1]},\cdots,\{y_{m,\pi_m(j)}\}_{j\in [n_m]}$ in the real protocol execution; here $\mathcal{S}$ doesn't know the involved permutations $\{\pi_i\}_{i \in [m]}$, except the information from the leakage profile $\mathcal{L}$. 
$\mathcal{S}$ performs simulation as follows. 
\begin{itemize}
        \item 
        For a \textit{corrupted} party $u_{i,j}$, $\mathcal{S}$ generates a ciphertext $(pk_{i,j},$ $\textsf{Enc}_{pk_{i,j}}(z_{i,\pi^{-1}_i(j)}))$ that will be placed to the $\pi^{-1}_i(j)$-th encrypted output in group $G_i$. 
    
        \item 
        For all other \textit{non-corrupted} parties, $\mathcal{S}$ randomly selects permutations $\{\widetilde{\pi}_i\}_{i \in [m]}$ under the constraint of the leakage profile $\mathcal{L}$, which means that $\widetilde{\pi}_i(j) = \pi_i(j)$ for each corrupted party $u_{i,j}$, where $\{\pi_i\}_{i \in [m]}$ is the collection of secret permutations used for shuffling in the real protocol execution. 
        Using $\{\widetilde{\pi}_i\}_{i \in [m]}$ and honest parties' public keys, $\mathcal{S}$ can encrypt the computation results corresponding to the honest parties accordingly. 
        Specifically, $\mathcal{S}$ generates $(pk_{i,j},\textsf{Enc}_{pk_{i,j}}(z_{i,\widetilde{\pi}^{-1}_i(j)}))$ that will be placed as the $\widetilde{\pi}^{-1}_i(j)$-th encrypted output in group $G_i$. 

    \end{itemize}
    
\item \label{2:step:3}  
To simulate the view corresponding to the shuffling phase, $\mathcal{S}$ first obtains (shuffled) input lists $L$; these shuffled lists are revealed to the server in real protocol execution, and here $\mathcal{S}$ obtains the information from the ideal functionality $\mathcal{F}_{\scriptscriptstyle \textsf{PIC}}$. 
Then, for the $j$-th client $u_{i,j}$ in group $G_i$, to simulate the input ciphertexts before the shuffling, the simulator $\mathcal{S}$ does the following: 
\begin{itemize} 
    \item If $u_{i,j}$ is corrupted, $\mathcal{S}$ gets the sanitized, shuffled input $x_{i,\pi_i^{-1}(j)}'$ from $L$. Then $x_{i,\pi_i^{-1}(j)}'$ is concatenated with $u_{i,j}$'s public key $pk_{i,j}$, and encrypted with the server's public key $x_{i,j}''=\textsf{Enc}_{pk_c}(pk_{i,j}||x'_{i,\pi_i^{-1}(j)})$. 
    
    \item If $u_{i,j}$ is non-corrupted, $\mathcal{S}$ gets the sanitized, shuffled input $x_{i,\widetilde{\pi}_i^{-1}(j)}'$ from $L$. Then $x_{i,\widetilde{\pi}_i^{-1}(j)}'$ is concatenated with $u_{i,j}$'s public key $pk_{i,j}$, and encrypted with the server's public key $x_{i,j}''=\textsf{Enc}_{pk_c}(pk_{i,j}||x'_{i,\widetilde{\pi}_i^{-1}(j)})$.   
\end{itemize} 
Clearly, $\mathcal{S}$ uses $\{\widetilde{\pi}_i\}_{i \in [m]}$ to arrange these ciphertexts in the above simulation. 
The generated ciphertexts above serve as the inputting ciphertexts for the shuffling phase. 
To simulate the output ciphertexts after the shuffling, $\mathcal{S}$ simply permutes all inputs ciphertexts using $\{\widetilde{\pi}_i\}_{i \in [m]}$, and re-randomizes each ciphertext. 
$\mathcal{S}$ then invokes the simulator $\mathcal{S}_{\scriptscriptstyle \textsf{Shuffle}}$ over the inputting ciphertexts, the outputting ciphertexts, and the leakage profile $\mathcal{L}$. 
$\mathcal{S}$ appends the generated view from $\mathcal{S}_{\scriptscriptstyle \textsf{Shuffle}}$ as a part of its own simulation. 

\end{enumerate} 

Below, we show the simulated view is indistinguishable from real protocol execution via the following hybrid games. 

$\mathbf{G_1}$: This is the real protocol execution. 

$\mathbf{G_2}$: $\mathbf{G_2}$ is same as $\mathbf{G_1}$, except the following difference:
    $\mathbf{G_2}$ uses randomly sampled permutations $\{\widetilde{\pi}_i\}_{i \in [m]}$ with the constraint that $\widetilde{\pi}_i(j) = \pi_i(j)$ for a corrupted party $u_{i,j}$ for $i \in [m]$. 
    The security of secure shuffling ensures that the difference from $\widetilde{\pi}_i(j) \neq \pi_i(j)$ for any non-corrupted party $u_{i,j}$ in $\mathbf{G}_2$ is indistinguishable from~$\mathbf{G}_1$. 

$\mathbf{G_3}$:
    $\mathbf{G_3}$ uses the function output $f(L)$ and $\{\widetilde{\pi}_i\}_{i \in [m]}$ to generated the final encrypted output for each party. For all other parts, $\mathbf{G_3}$ remains the same as $\mathbf{G_2}$. 
    In particular, $\mathbf{G_3}$ encrypts the desired outputs for all parties using $f(L)$ and $\{\widetilde{\pi}_i\}_{i \in [m]} \}$, as done in Step (\ref{2:step:2}) of $\mathcal{S}$. 
    $\mathbf{G_3}$ is perfectly indistinguishable from $\mathbf{G_2}$. 

$\mathbf{G_4}$:
    $\mathbf{G_4}$ uses the shuffled output $L$ and $\{\widetilde{\pi}_i\}_{i \in [m]}$ to generated ciphertexts before and after the shuffling. 
    For all other parts, $\mathbf{G_4}$ remains the same as $\mathbf{G_3}$. 
    In particular, $\mathbf{G_4}$ encrypts the desired outputs for all parties using $L$ and $\{\widetilde{\pi}_i\}_{i \in [m]} \}$, as done in Step (\ref{2:step:3}) of $\mathcal{S}$. 
    $\mathbf{G_4}$ is perfectly indistinguishable from $\mathbf{G_3}$. 

\textbf{$\mathbf{G_5}$:} 
    The only difference between $\mathbf{G_5}$ and $\mathbf{G_4}$ is that $\mathbf{G_5}$ invoke $\mathcal{S}_{\scriptscriptstyle \textsf{Shuffle}}$ to simulate the view for the shuffling phase. 
    Note that $\mathbf{G_5}$ works exactly as the simulator $\mathcal{S}$. 
    $\mathbf{G_5}$ is perfectly indistinguishable from $\mathbf{G_4}$.  
    
Overall, the view generated by $\mathcal{S}$ is perfectly indistinguishable from the view of real protocol execution in the $\mathcal{F}_{\scriptscriptstyle \textsf{Shuffle}}$-hybrid model.  
Intuitively, the only secret information when $S$ is corrupted is the hidden part of the permutations $\{\pi_i\}_{i\in [m]}$ corresponding to each honest party $u_{i,j}$ for all $i \in [m]$ and $j \in [|\mathbf{G}_j|]$, which can be perfectly simulated in the $\mathcal{F}_{\scriptscriptstyle \textsf{Shuffle}}$-hybrid model. 

\section{Privacy Proof for Theorem \ref{the:privacy}}\label{app:privacyproof}

Our proof frequently uses the data processing inequality of DP. The data processing inequality is considered a key feature of distance measures (e.g., Hockey-stick divergence) used for evaluating privacy. It asserts that the privacy guarantee cannot be weakened by further analysis of a mechanism's output.

\begin{definition}[Data processing inequality]\label{def:postprocess}
A distance measure $D:\Delta(\mathcal{S})\times\Delta(\mathcal{S})\to[0,\infty]$ on the space of probability distributions satisfies the data processing inequality if, for all distributions $P$ and $Q$ in $\Delta(\mathcal{S})$ and for all (possibly randomized) functions $g:\mathcal{S}\to\mathcal{S'}$, 
\[D(g(P)\|g(Q))\le D(P\|Q).\]
\end{definition}

The situations of victim user $u_{i^*,j^*}$'s privacy can be categorized into three case: (1) $u_{i^*,j^*}$ himself/herself is corrupted; (2) $u_{i^*,j^*}$ is not corrupted and the server $S$ is a corrupted party (i.e., $S\in C$); (3) neither $u_{i^*,j^*}$ nor the server $S$ is corrupted. In the first case, protecting the privacy of $u_{i^*,j^*}$ against himself/herself is trivial. In the second case, a key observation is that the received computation results of all corrupted parties (including $S$) in $C$ are rooted from the server's received messages $\{x_{i,\pi_(j)}''\}_{i\in[m],j\in [n_i]}$ in Step (4). We let variable $Y'$ denote the server's received messages $\{x_{i,\pi_(j)}''\}_{i\in[m],j\in [n_i]}$ when the input datasets are $X'$, and let variable $Y$ denote the received messages when the input datasets are $X$. Then, for any distance measure $D$ that satisfies the data processing inequality, we readily have the divergence level of the adversaries' view (i.e. $\mech{A}(X)$ when the input datasets are $X$, and $\mech{A}(X')$ when the input datasets are $X'$) is upper bounded by the one about the server's received messages:
\begin{equation}\label{eq:a2}
D(\mech{A}(X)\|\mech{A}(X'))\leq D(Y\| Y').
\end{equation}

We now focus on the DP guarantee of the server's received messages in Step 4 (i.e. the divergence $D(Y\| Y')$). For neighboring datasets $X$ and $X'$ that differ at (and only at) the user $u_{i^*,j^*}$ from group $i^*$, we have the $D(Y\| Y')$ upper bounded by the divergence of shuffled messages from \emph{uncorrupted users} in group $X_{i^*}$. To be precise, we let $\hat{Y}_{i^*}$ denote the shuffled messages $\mech{S}(\{x_{i^*,j}''\}_{u_{i^*,j}\in G_{i*}\backslash C_{i^*}})$ of group $i^*$ when the input dataset is $X$, and let $\hat{Y}'_{i^*}$ denote the shuffled messages $\mech{S}(\{x_{i^*,j}''\}_{u_{i^*,j}\in G_{i*}\backslash C_{i^*}})$ when the input dataset is $X'$. Then for any distance measure that satisfies the data processing inequality, we have $D(Y\| Y')\leq D(\hat{Y}_{i^*}\| \hat{Y}'_{i^*})$ (see Lemma \ref{lemma:y2yi} for proof).

\begin{lemma}\label{lemma:y2yi}
For neighboring datasets $X$ and $X'$ that differ at (and only at) the user $u_{i^*,j^*}$ from group $i^*$, we have:
$$D(Y\| Y')\leq D(\hat{Y}_{i^*}\| \hat{Y}'_{i^*}).$$
\end{lemma}
\begin{proof}
Without loss of generality (for both corrupted/malicious and uncorrupted users), we let $\mech{R}_{i,j}$ denote the (possibly random) local message generation function of user $u_{i,j}$ in Step (2), which takes as input the local information (i.e. $x_{i,j}$, $pk_{i,j}$, $sk_{i,j}$) and global parameters. To prove the result, we define the following procedures (denoted as $post$) that takes as input the shuffled messages $\mech{S}(\{x_{i^*,j}''\}_{u_{i^*,j}\in G_{i^*}\backslash C_{i^*}})$ and outputs $\mech{S}(\{x_{i,j}''\}_{i\in [m], u_{i,j}\in G_i})$.
\begin{itemize}
    \item For $u_{i^*,j}\in C_{i^*}$, running $\mech{R}_{i^*,j}$ to obtain $x_{i^*,j}''$;
    \item Uniformly sample a permutation $\pi'_{i^*}:C_{i^*}\mapsto [n_{i^*}]$, then construct a $n_{i^*}$-length list $Y_{i^*}$ such that every $x_{i^*,j}''$ (for corrupted users $u_{i^*,j}\in C_{i^*}$) asides at the $\pi'_{i^*}(j)$-th position, and then inputted messages $\mech{S}(\{x_{i^*,j}''\}_{u_{i^*,j}\in G_{i^*}\backslash C_{i^*}})$ sequentially fill up the list;  
    \item For all $i\in [m]\backslash \{i^*\}$ and all possible $j\in G_i$, running $\mech{R}_{i,j}$ to obtain $x_{i,j}''$;
    \item For all $i\in [m]\backslash \{i^*\}$, uniformly sample a permutation $\pi_{i}:G_i\mapsto [n_i]$ and get $\{x_{i,\pi_{i}(j)}''\}_{j\in [n_i]}$;
    \item Initialize $Y_{temp}$ as $Y_{i^*}$. For $i\in [i^*-1]$, prepend $\{x_{i,\pi_{i}(j)}''\}_{j\in [n_i]}$ to $Y_{temp}$; for $i\in [i^*+1,m]$, append $\{x_{i,\pi_{i}(j)}''\}_{j\in [n_i]}$ to $Y_{temp}$;
    \item Return $Y_{temp}$.
\end{itemize}
It is oblivious that: (1) when the input dataset is $X$, the $post(\hat{Y}_i)$ distributionally equals to $Y$; (2) when the input dataset $X'$, the $post(\hat{Y}_i')$ distributionally equals to $Y'$. Therefore, according to the data processing inequality, we have the conclusion.
\end{proof}


We proceed to analyze the DP guarantee of shuffled messages from uncorrupted users in group $i^*$ (i.e. the divergence $D(\hat{Y}_{i^*})\| \hat{Y}'_{i^*})$). To this end, we firstly summarize the local message generating procedures aside on each uncorrupted client's side as an abstracted mechanism, then show the extended \& encrypted message outputted by the abstracted mechanism has the same local privacy guarantee as the data randomizer $\mech{R}$, and show all clients in the same group follow an identical abstracted mechanism. Finally, we show the privacy amplification via shuffling can be applied to this identical local mechanism, as if purely shuffling data randomized by $\mech{R}$.

\textbf{An abstraction for local client-side procedures. } Recall that the only message leaving a uncorrupted client $j$ (i.e. $u_{i^*,j}\in G_{i^*}\backslash C_{i^*}$) in Phases (0)-(4) is (fingerprint/public key $f_{i^*,j}$, sanitized data $x'_{i^*,j}$, extra data $t_{i^*,j}$), where the fingerprint $f_{i^*,j}=F(sk_{i^*,j})$ is post-processed from the random private key $sk_{i^*,j}$ using some function $F$, the sanitized data $x'_{i^*,j}$ is the privatized version of secret data $x_{i^*,j}=((i^*,j), l_{i^*,j}, v_{i^*,j})$ via a privatization mechanism $\mech{R}$, and the extra data $t_{i^*,j}$ is post-processed as $T(sk_{i^*,j}, f_{i^*,j}, x'_{i^*,j})$ using some function $T$. We summarize and abstract these procedures in Algorithm \ref{alg:local}.

\begin{algorithm}[t]
\small
    \SetKwInOut{Parameter}{Params}
    \renewcommand\baselinestretch{1.0}\selectfont
    \caption{An abstracted local mechanism $\overline{\mech{R}}$}
    \label{alg:local}
    \Parameter{A domain $\dom{D}_{k}$, a distribution $P_{k}:\dom{D}_{k} \mapsto [0,1]$, any function $F:\dom{D}_{k}\mapsto\dom{D}_{f}$, any local mechanism $\mech{R}: \dom{X}\mapsto \dom{Y}$, and any function $T:\dom{D}_{k}\times\dom{D}_{f}\times\dom{Y}\mapsto \dom{D}_{z}$.}
    \KwIn{An input $x_{i^*,j}\in \dom{X}$.}
    \KwOut{An extended message in Phase (1).}

    {$x'_{i^*,j} \leftarrow \mech{R}(x_{i^*,j})$}
    
    {sample $sk_{i^*,j} \sim P_{k}$}

    {$f_{i^*,j} \leftarrow F(sk_{i^*,j})$}

    {$t_{i^*,j} \leftarrow T(sk_{i^*,j}, f_{i^*,j}, y_{i^*,j})$}
    
    \KwRet{$(f_{i^*,j},\ x'_{i^*,j},\ t_{i^*,j})$}
\end{algorithm}

\textbf{Identicalness of local mechanisms. } Recall that privacy amplification via shuffling requires, not only that each local data is randomized, but also that clients follow an identical randomization mechanism (otherwise anonymity might break due to the output domain/distribution difference in different randomizers). In Lemma \ref{lemma:1ldp}, we show that the messages from every \emph{honest} client (in the same group) are indeed randomized by an identical mechanism (i.e., the Algorithm \ref{alg:local}). The proof of the lemma is trivial since every honest client in the same group adopted the same global parameters and the same local randomizer since Phase (0).

\begin{lemma}\label{lemma:1ldp}
For each honest/uncorrupted client in the same group, Phase (1) in the conceptual protocol from Section \ref{sec:protocol} is equivalent to Algorithm \ref{alg:local} with global parameters given in Phase (0).
\end{lemma}

\textbf{Privacy amplification guarantees. } Since every message from clients in the same group is sanitized by an identical mechanism $\overline{\mech{R}}$ (see the former part), and the server (i.e., the potential privacy adversary) only observe the shuffled messages from each group, one can directly apply the amplification bounds in the literature \cite{balle2019privacy,feldman2021hiding,feldman2022stronger,wang2023unified}. 
Generally, for any distance measure $D$ that satisfies the data processing inequality,  we have the corresponding distance between $\hat{Y}_{i^*})$ and $\hat{Y}'_{i^*})$, is upper bounded by the distance between $\mech{S}\circ\mech{R}(\{x_{i^*,j}\}_{u_{i,j}\in G_{i^*}\backslash C_{i^*}})$ and $\mech{S}\circ\mech{R}(\{x'_{i^*,j}\}_{u_{i,j}\in G_{i^*}\backslash C_{i^*}})$ (see Lemma \ref{lemma:shuffle1ldp}).

\begin{lemma}[Privacy Guarantee of Shuffled Algorithm \ref{alg:local}]\label{lemma:shuffle1ldp}
Given two neighboring dataset $\hat{X}_{i^*}=\{x_{i^*,j}\}_{u_{i,j}\in G_{i^*}\backslash C_{i^*}}$ and  $\hat{X}'_{i^*}=\{x'_{i^*,j}\}_{u_{i,j}\in G_{i^*}\backslash C_{i^*}}$ that differ only at one element, then for any distance measure $D$ that satisfies the data processing inequality,
\[D(\mech{S}\circ\overline{\mech{R}}(\hat{X}_{i^*})\|\mech{S}\circ\overline{\mech{R}}(\hat{X}'_{i^*}))\leq D(\mech{S}\circ\mech{R}(\hat{X}_{i^*})\| \mech{S}\circ\mech{R}(\hat{X}'_{i^*})).\]
\end{lemma}
\begin{proof}
Given input either $\mech{S}\circ\mech{R}(\hat{X}_{i^*})$ or $\mech{S}\circ\mech{R}(\hat{X}'_{i^*})$, we define a following function (denoted as $g_{sr}$):
\begin{itemize}

\item For each message $x'_{i^*,j}$ in the input, sample $sk_{i^*,j}\sim P_k$, compute $f_{i^*,j}\leftarrow F(sk_{i^*,j})$, then compute $t_{i^*,j}\leftarrow T(sk_{i^*,j},f_{i^*,j},x'_{i^*,j})$, and finally get $(f_{i^*,j},x'_{i^*,j},t_{i^*,j})$.

\item Return a sequential of $(f_{i^*,j},x'_{i^*,j},t_{i^*,j})$ each is generated by the previous step over each input message. 
\end{itemize}
It is obvious that the $g_{sr}(\mech{S}\circ\mech{R}(\hat{X}_{i*}))$ distributionally equals to $\mech{S}\circ\overline{\mech{R}}(\hat{X}_{i*})$, and $g_{sr}(\mech{S}\circ\mech{R}(\hat{X}'_{i*}))$ distributionally equals to $\mech{S}\circ\overline{\mech{R}}(\hat{X}'_{i*})$. Therefore, according to the data processing inequality of the distance measure of $D$, we have the conclusion.
\end{proof}

In particular, when the mechanism $\mech{R}$ satisfies $\epsilon$-LDP, we have $\overline{\mech{R}}$ satisfies $\epsilon$-LDP. Then according to the latest work \cite{feldman2022stronger},  the shuffled messages $\mathcal{S}(\{\overline{\mech{R}}(x_{i^*,j})\}_{u_{i,j}\in G_{i^*}\backslash C_{i^*}})$ from group $G_{i^*}$ satisfy $(\epsilon_c,\delta)$-DP where \[\epsilon_c=\texttt{Amplify}(\epsilon,\delta,|G_{i^*}\backslash C_{i^*}|)=\tilde{O}(\sqrt{e^{\epsilon}/n'_{i^*}}).\]
Specifically, when $n\geq 8 (e^{\epsilon}+1)\log(2/\delta)$, we have $\epsilon_c\leq \log\Big(1+\frac{e^\epsilon-1}{e^\epsilon+1}(\sqrt{\frac{64 e^{\epsilon} \log 4/\delta}{n'_{i^*}}}+\frac{8 e^\epsilon}{n'_{i^*}})\Big)$ \cite[Theorem 3.1]{feldman2021hiding}. This implies the conclusion of the second case.

Finally, let us consider the third case. Simply add the server $S$ to corrupted parties $C$ (i.e. more views are leaked to adversaries), and apply the result for the second case, we arrive at the conclusion.

\section{Proof of Error Lower Bounds in PIC Model}\label{app:errorlower}

We firstly shift our focus from the original $\ell_2$-norm spherical domain $\mathbb{S}_{2,1}({0}^d)$ to the simpler $\ell_{+\infty}$-norm domain $[0,1]^d$. To establish the lower bound of the mean square error (MSE) of a single report derived from the single-message shuffle model with $x_i\in [0,1]^d$ as input, we follow a four-step approach. First, we confine the space of local randomizers to the subspace where the domain of the randomizer's output matches the input. Second, we construct a set of discretized and gridded inputs that are well separated. Third, we compute their expected MSE using a formula that hinges on the probability transition matrix among these discrete inputs. Finally, we leverage the constraints of differential privacy inherent to the shuffle model to ascertain the properties of the probability transition matrix, thereby deducing lower bounds.

\textbf{Step (1):} We begin by demonstrating that the MSE bound of any local randomizer $R:[0,1]^d\mapsto \mathbb{R}^{d'}$ coupled with any post-processing function $f:\mathbb{R}^{d'}\mapsto \mathbb{R}^d$ is lower bounded by a certain local randomizer $R':[0,1]^d\mapsto [0,1]^d$. This local randomizer, $R'$, possesses an output domain identical to the input. We establish it using a constructive method. Defining $R'(x_i) = \max(0, \min(1, f(R(x_i)))$, where $\min$ and $\max$ execute coordinate-wise min/max truncation, we can say for any $x_i\in [0,1]^d$ that:
\begin{alignat*}{2}
&\mathbb{E}[\|f(R(x_i))-x_i\|_2^2]\\
=&\int_{x\in \mathbb{R}^d} \mathbb{P}[f(R(x_i))=x]\cdot \|x-x_i\|_2^2 \mathrm{d}x\\
\geq&\int_{x\in \mathbb{R}^d}  \mathbb{P}[f(R(x_i))=x]\cdot \|\max(0, \min(1, x))-x_i\|_2^2 \mathrm{d}x\\
\geq&\mathbb{E}[\|R'(x_i)-x_i\|_2^2].
\end{alignat*}
As a result, we can now narrow our focus to local randomizers with an output domain of $[0,1]^d$.

\textbf{Step (2):} We proceed by partitioning the domain $[0,1]^d$ into multi-dimensional grids, such that the center points of these grids are well-separated (i.e., have non-zero distances from each other). Specifically, each dimension is segmented into $L$ uniform intervals, where the $l$-th interval is defined as $[\frac{l-1}{L}, \frac{l}{L})$ for $l\in [L-1]$, and the $L$-th interval is $[\frac{L-1}{L}, 1]$. Intervals across all dimensions divide the domain into grids or subdomains, yielding a total of $L^d$ grids. Each grid is indexed by the indices of its intervals in each dimension. For instance, the $(1,1,\ldots,1)$-th grid corresponds to the subdomain $[0, \frac{1}{L})\times [0, \frac{1}{L})\times \cdots \times[0, \frac{1}{L}) \in [0,1]^d$. Each grid possesses a center point, with the center point of the $(l_1,l_2,\ldots,l_d)$-th grid being $\big[\frac{l_1-1/2}{L},\frac{l_2-1/2}{L},\ldots,\frac{l_d-1/2}{L}\big]$ for $l_1,\dots,l_d\in [L]^d$. We denote all center points as:
\[C=\Big\{\big[\frac{l_1-1/2}{L},\frac{l_2-1/2}{L},\ldots,\frac{l_d-1/2}{L}\big]\ |\ l_1,\dots,l_d\in [L]^d\Big\}.\]
We use $G(l_1,l_2,\ldots,l_d)$ to denote the subdomain of the $(l_1,l_2,\ldots,l_d)$-th grid.

\textbf{Step (3): } Following the outcome of Step (1), we examine any local randomizer $R:[0,1]^d\mapsto[0,1]^d$. Given any input $x_i\in C$, the randomizer $R$ defines a probabilistic transition from $x_i$ to each of the $L^d$ grids. We represent the transition probability from $x_i$ to the $\mathbf{l'}$-th grid as $P_{\mathbf{l},\mathbf{l'}}$, where $\mathbf{l}$ is the index of the grid that contains $x_i$ and $\mathbf{l},\mathbf{l'} \in [L]^d$. By iterating over all possible $x_i \in C$, we obtain a transition probability matrix. Now, considering that $x$ is randomly sampled from $C$ in a uniform manner, our objective is to analyze the \emph{expected} MSE of $R$ given $x$ as input. Given that a center point maintains a minimum Manhattan distance of $\frac{1}{2L}$ from other grids, we can lower bound the expected MSE as follows:
\begin{alignat*}{2}
 &\mathbb{E}_{x\sim uniform(C)}[\|R(x)-x\|_2^2]\\
\geq &\frac{1}{L^d}\sum_{\mathbf{l}\in [L]^d} \sum_{\mathbf{l'}\in [L]^d} P_{\mathbf{l},\mathbf{l'}}\cdot
\llbracket \mathbf{l}\neq \mathbf{l'}\rrbracket \cdot \frac{1}{(2L)^2} \\ 
\geq&\frac{1}{L^d}\sum_{\mathbf{l}\in [L]^d}\frac{1- P_{\mathbf{l},\mathbf{l}}}{4 L^2},
\end{alignat*}
where $P_{\mathbf{l},\mathbf{l}}$ represents the probability that the output resides within the same interval as the input central point of the $\mathbf{l}$-th grid.
Additionally, since the squared distance between two points from two different grids (e.g., from the $\mathbf{l}$-th and $\mathbf{l'}$-th grid respectively) is at least $\frac{1}{L^2}\sum_{j\in [d]}\llbracket \mathbf{l}_j\neq \mathbf{l'}_j \rrbracket(|\mathbf{l}_j-\mathbf{l'}_j|-1/2)^2$, we lower bound the expected MSE as follows:
{\small
\begin{alignat*}{2}
 &\mathbb{E}_{x\sim uniform(C)}[\|R(x)-x\|_2^2]\\
\geq&\frac{1}{L^d}\sum_{\mathbf{l}'\in [L]^d} \sum_{\mathbf{l}\in [L]^d}  P_{\mathbf{l'},\mathbf{l}}\cdot \frac{\sum_{j\in [d]}\llbracket \mathbf{l}_j\neq \mathbf{l'}_j \rrbracket(|\mathbf{l'}_j-\mathbf{l}_j|-1/2)^2}{L^2}\\
\geq&\frac{1}{L^d}\sum_{\mathbf{l}'\in [L]^d}\big(\min_{\mathbf{l''}\in [L]^d}  P_{\mathbf{l''},\mathbf{l}}\big)\sum_{\mathbf{l}\in [L]^d} \frac{\sum_{j\in [d]}\llbracket \mathbf{l}_j\neq \mathbf{l'}_j \rrbracket(|\mathbf{l'}_j-\mathbf{l}_j|-1/2)^2}{L^2}\\
\geq&\frac{1}{L^d}\sum_{\mathbf{l}'\in [L]^d}\big(\min_{\mathbf{l''}\in [L]^d}  P_{\mathbf{l''},\mathbf{l}}\big)\sum_{j\in [d]}\sum_{\mathbf{l}\in [L]^d} \frac{\llbracket \mathbf{l}_j\neq \mathbf{l'}_j \rrbracket(|\mathbf{l'}_j-\mathbf{l}_j|-1/2)^2}{L^2}\\
\geq&\frac{1}{L^d}\sum_{\mathbf{l}'\in [L]^d} \big(\min_{\mathbf{l''}\in [L]^d}  P_{\mathbf{l''},\mathbf{l}}\big)\sum_{j\in [d]} \frac{L^{d-1}(L-1/L)}{48},
\end{alignat*}
}
where the last step uses the fact that $ \sum_{\mathbf{l}\in [L]^d} \frac{\llbracket \mathbf{l}_j\neq \mathbf{l'}_j \rrbracket(|\mathbf{l'}_j-\mathbf{l}_j|-1/2)^2}{L^2}\geq \frac{L^{d-1}(L-1/L)}{48}$ holds for all possible $ \mathbf{l'}$. The $\min_{\mathbf{l''}\in [L]^d}  P_{\mathbf{l''},\mathbf{l}}$ denotes the minimum possible probability that the output falls within the same interval as the central point of the $\mathbf{l}$-th grid, given all possible input $\mathbf{l''}\in [L]^d$.



\textbf{Step (4): } We now leverage the DP constraint in the shuffle model to establish a relationship between the following two probabilities:
\begin{alignat*}{1}
p_1=P_{\mathbf{l},\mathbf{l}},\\ p_0=P_{\mathbf{l''},\mathbf{l}}
\end{alignat*}
when $\mathbf{l},\mathbf{l''}$. Let's consider two central points $x,x''\in C$ such that $x$ and $x''$ belong to the $\mathbf{l}$-th and $\mathbf{l''}$-th grid respectively. We can then construct two neighboring datasets $T=(x'',x'',\ldots,x'')$ and $T''=(x,x'',\ldots,x'')$ that both contain $n$ elements. Then, in two independent runs of $S\circ R(T)$ and $S\circ R(T'')$, the corresponding probability that $S\circ R(T)$ (or $S\circ R(T'')$) contains no elements within the $\mathbf{l}_j$-th interval, is constrained by $(\epsilon,\delta)$-differential privacy as follows (assuming $\delta<0.5$):
\begin{alignat}{1}\label{eq:dpineq}
\mathbb{P}[S\circ R(T)\cap G(\mathbf{l})=\emptyset]\leq e^{\epsilon}\mathbb{P}[S\circ R(T'')\cap G(\mathbf{l})=\emptyset]+\delta,
\end{alignat}
where $G(\mathbf{l})\subseteq [0,1]^d$ denotes the domain of the $\mathbf{l}$-th grid. Note that $\mathbb{P}[S\circ R(T)\cap G(\mathbf{l})=\emptyset]=(1-p_0)^n$ and $\mathbb{P}[S\circ R(T)\cap G(\mathbf{l})=\emptyset]=(1-p_1)(1-p_0)^{n-1}$. Consequently, if $(1-p_1)\leq 0.5 e^{-\epsilon}$, then we have $p_0\geq (1/2-\delta)/n$ where $\delta<1/2$, since $(1-p_0)^n> 1/2+\delta> e^{\epsilon}(1-p_1)+\delta\geq 1/2+\delta$ results in a contradiction \cite[Lemma 4.5]{balle2019privacy}. This implies that either $(1-p_1)> 0.5 e^{-\epsilon}$ or $p_0\geq (1/2-\delta)/n$ holds for all possible $\mathbf{l''},\mathbf{l}\in [L]^d$ (under the assumption that $\delta<1/2$). Therefore, we arrive at the expected MSE as follows:
\begin{alignat*}{2}
 &\mathbb{E}_{x\sim uniform(C)}[\|R(x)-x\|_2^2]\\
\geq&\frac{1}{L^d}\sum_{\mathbf{l}\in [L]^d}\min\Big\{\frac{1-p_1}{4L^2},\ \frac{d L^{d-1}(L-1/L)p_0}{48}  \Big\}\\
\geq&  \frac{1}{L^d}\sum_{\mathbf{l}\in [L]^d}\min\Big\{\frac{e^{-\epsilon}}{8 L^2},\ \frac{1/2-\delta}{n}\cdot  \frac{d L^{d-1}(L-1/L)}{96}\Big\}
\end{alignat*}
Choosing $L$ at $\lceil {(n/d)}^{1/(d+2)} \rceil$ yields the expected MSE as $\tilde{\Omega}(\frac{d^{2/(d+2)}}{n^{2/(d+2)}})$. As we are concerned with the asymptotic MSE with respect to $n$, the parameters $(\epsilon,\delta)$ are suppressed in the bound.

Because the expected MSE bounds will never exceed the worst-case MSE bounds, we can establish that for an input data domain $\dom{X}=[0, 1]^d$, the MSE lower bound $\max_{x\in \dom{X}} R(x)\geq \tilde{\Omega}(\frac{d^{2/(d+2)}}{n^{2/(d+2)}})$. Then, for the re-scaled domain $\mathbb{X}'=[0, {1}/{\sqrt{d}}]^d$, the MSE lower bound is $\max_{x'\in \dom{X}'} R(x')\geq \tilde{\Omega}(\frac{d^{-d/(d+2)}}{n^{2/(d+2)}})$. Finally, using the fact that $[0, {1}/{\sqrt{d}}]^d \subseteq \mathbb{B}_{2,1}(\{0\}^{d})$, we can conclude that the MSE lower bound over the domain $\mathbb{B}_{2,1}(\{0\}^{d})$ is $\tilde{\Omega}(\frac{1}{n^{2/(d+2)}})$ when $d$ is fixed and $n$ is sufficiently large. For the $d=1$ case, the \cite{balle2019privacy} has established the $\tilde{\Omega}(\frac{1}{n^{2/3}})$ bound; the concurrent work in \cite{asiprivate2024} also extends  to multi-dimension case for the hyperspherical unit domain and obtains similar results.


\end{document}